\documentclass[11pt, onecolumn,  
journal]{IEEEtran} \usepackage{amsmath}
\usepackage{graphicx,amssymb}
\usepackage{algorithm}
\usepackage{algorithmic}
\usepackage{amsthm,color}
\usepackage{diagbox}
\usepackage{multirow}
\usepackage{subcaption}
\usepackage{soul}
\usepackage{yhmath}
\usepackage{mathdots}
\usepackage{MnSymbol}

\makeatletter

\@addtoreset{algorithm}{section}
\makeatother

\newtheorem{theorem}{Theorem}[section]

\newtheorem{proposition}[theorem]{Proposition}
\newtheorem{definition}[theorem]{Definition}

\newtheorem{remark}[theorem]{Remark}

\def\G{ {\mathcal G} }
\def\R{ {\mathbb R} }

\def\V{ {\bf V} }

\def\L{ {\bf L} }
\def\U{ {\bf U} }

\def\A{ {\bf A} }

\numberwithin{equation}{section}

\begin{document}

\title{Graph Fourier transform based on singular value decomposition of  directed Laplacian }
\author{Yang Chen, Cheng Cheng, Qiyu Sun
\thanks{Chen is with Key Laboratory of Computing and Stochastic Mathematics (Ministry of Education), School of Mathematics and Statistics,
 Hunan Normal University, Changsha, Hunan 410081, China;
Cheng is with  School of Mathematics, Sun Yat-sen University,   Guangzhou, Guangdong 510275, China;
Sun is with  Department of Mathematics, University of Central Florida, Orlando, Florida 32816, USA;
Emails: ychenmath@hunnu.edu.cn; chengch66@mail.sysu.edu.cn;   qiyu.sun@ucf.edu;
This work is partially supported by
the  National Science Foundation (DMS-1816313),  National Nature Science Foundation of China (11901192, 12171490), Guangdong Province Nature Science Foundation (2022A1515011060), and  Scientific Research Fund of Hunan Provincial Education Department(18C0059).
}}
\maketitle

\begin{abstract}
Graph Fourier transform (GFT) is a fundamental concept in graph signal processing.
In this paper, based on singular value decomposition of Laplacian,
we introduce a novel definition of  GFT
on directed graphs, and use  singular values of Laplacian to carry the notion of graph frequencies. 
The proposed GFT is consistent with the conventional GFT in the undirected graph setting,
and on directed circulant graphs,
the proposed GFT is the classical discrete Fourier transform, up to some rotation, permutation and phase adjustment. We  show that
  frequencies and frequency components of the proposed GFT can be evaluated by solving some constrained minimization problems with  low computational cost.
Numerical demonstrations indicate that the proposed GFT could
 represent graph signals with different modes of variation  efficiently.

\end{abstract}

\section{Introduction}

Graph signal processing provides
an innovative framework to represent, analyze and process data sets residing on networks, and its mathematical foundation is closely related to   applied and computational harmonic analysis and  spectral graph theory
\cite{sandryhaila13}-\cite{Cheng19}.
Graph Fourier transform (GFT)
is one of fundamental tools  in graph signal processing that
 decomposes  graph signals  into different frequency components and
 represents them by different modes of variation.
 GFT on directed graphs is an important tool to
 identify patterns and quantify influence of various members and communities of a social network,  and to understand  dynamic of a network.
 The GFT on undirected graphs has been well-studied and several approaches have been proposed to define GFT on directed graphs
 \cite{aliaksei14, stankovic2019introduction, Chungbook}-\cite{Yang21}.
  In this paper,  we introduce a novel definition of  GFT
on directed graphs, which is based on singular value decomposition of the associated Laplacian, see Definition \ref{fourier.def}.

Let $\G=(V, E)$ be  a weighted (un)directed graph of order $N$
containing no loops or
multiple edges, and denote the associated adjacency matrix, in-degree matrix and Laplacian by ${\bf A}, {\bf D}$ and ${\bf L}:={\bf D}-{\bf A}$ respectively.
In the undirected graph setting (i.e., the associated adjacent matrix ${\bf A}$ is symmetric),
the  Laplacian ${\mathbf L}$ is positive semi-definite  and it has  the following  eigendecomposition
\begin{equation} \label{undirectedLaplacian.def1}
{\bf L}= {\bf V} {\pmb \Lambda} {\bf V}^T=\sum_{i=0}^{N-1} \lambda_i {\bf v}_i {\bf v}_i^T,
\end{equation}
where ${\bf V}=[{\bf v}_0, \ldots, {\bf v}_{N-1}]$ is an orthogonal matrix and
${\pmb \Lambda}={\rm diag} (\lambda_0, \ldots, \lambda_{N-1})$ is a diagonal matrix of nonnegative eigenvalues of the Laplacian ${\mathbf L}$  in nondecreasing order, i.e.,
$
\lambda_0\le \ldots\le \lambda_{N-1}$. 
A well-accepted 
definition of GFT in the undirected  graph setting
is given by
\begin{equation} \label{undirectedLaplacian.def2}
{\mathcal F} {\bf x}= {\bf V}^T {\bf x}=\sum_{i=0}^{N-1} \langle {\bf x}, {\bf v}_i\rangle {\bf v}_i,
\end{equation}
where ${\bf x}$ is  a  graph signal and $\langle \cdot, \cdot\rangle $ 
 is the standard inner product on ${\mathbb R}^N$  \cite{Ricuad2019, stankovic2019introduction, Chungbook,  Lu2019, Yang21, magoarou2018}. 
The eigenvalues $\lambda_i$ and the associated eigenvectors ${\bf v}_i, 0\le i\le N-1$,
of the Laplacian ${\mathbf L}$ are considered as  frequencies and frequency components of the GFT  just defined. It is known that
the GFT in  \eqref{undirectedLaplacian.def2}
is orthogonal and on a cycle graph, it is  essentially the classical discrete Fourier transform.


The GFT  \eqref{undirectedLaplacian.def2}  does not apply directly for  weighted and directed graphs, which are widely used  to describe the interaction structure of a social network
that  has members of various types, such as individuals, organizations, leaders and followers, and the
pairwise interactions between members being not always mutual and  equitable
 \cite{Wasserman1994}-\cite{Segarra17}.
A natural approach is to replace the eigendecomposition
\eqref{undirectedLaplacian.def1} of the Laplacian ${\mathbf L}$ by the Jordan decomposition
$\L=\V {\bf J} \V^{-1}$, and then to define the GFT   of a  signal ${\bf x}$  on a directed graph by
\begin{equation}
\label{Jordan.gft.def}
{\cal F}{{\bf x}}=\V^{-1}{{\bf x}}
\end{equation}
\cite{aliaksei14,  sandryhaila2013, sandryhaila14, Singh16, deri2017, Domingos20}.
The  GFT in \eqref{Jordan.gft.def} could have complex frequencies and it is not always unitary. More  critically,
  Jordan decomposition of the Laplacian  ${\mathbf L}$ on directed graphs could be
numerically unstable and computationally  expensive, and hence it could be difficult to be applied for graph spectral analysis and decomposition,
see \cite{Domingos20} for a modified Jordan decomposition with some numerical stability.
Our GFT  in Definition \ref{fourier.def}
is based on  the  singular value decomposition (SVD)
\begin{equation}
\label{svd.def}\L=\U\pmb \Sigma\V^T=\sum_{i=0}^{N-1} \sigma_i {\bf u}_i {\bf  v}_i^T \end{equation}
of  the Laplacian  $\L$
and has its nonnegative singular values $\sigma_i, 0\le i
\le N-1$, as frequencies  and ${\bf u}_i, {\bf v}_i, 0\le i\le N-1$ as the associated left/right frequency components, where
\begin{equation}\label{UV.component}
 {\bf U}=[{\bf u}_0, \ldots, {\bf u}_{N-1}] \ {\rm and} \   {\bf V}=[{\bf v}_0, \ldots, {\bf v}_{N-1}]
\end{equation}
 are  orthogonal matrices, and the diagonal matrix
 $\pmb \Sigma={\rm diag} (\sigma_0, \ldots, \sigma_{N-1})$
   has   singular values  deployed on the diagonal in a nondecreasing order, i.e.,
$    0\le\sigma_0\le\sigma_1\le\ldots\le\sigma_{N-1}$.  
Compared  with the GFT \eqref{Jordan.gft.def} based on a Jordan decomposition, a significant advantage of the proposed SVD-based GFT is
on numerically stability and low computational cost.

  Given a graph signal ${\bf x}$ on a directed graph, denote its Euclidean norm   by $\|{\bf x}\|_2$, and define
 its quadratic variation  
by
 \begin{equation}\label{TV.def00}
 {\rm QV}({\bf x})  = {\bf x}^T\L{\bf x}=\frac{1}{2} {\bf x}^T(\L+\L^T){\bf x}
 \end{equation}
 \cite{Ricuad2019, stankovic2019introduction, Chungbook, Girault2018}.
 In the undirected graph setting,
 frequencies $\lambda_i$ and their corresponding frequency components
  ${\bf v}_i, 0\le i\le N-1$, of the GFT in \eqref{undirectedLaplacian.def2}
can be obtained via solving the following constrained minimization problems
         \begin{equation}\label{CFminimization.eq00}
\left\{\begin{array}{l}
\lambda_i=\min_{{\bf x}\in W_i^\perp \ {\rm with} \ \|{\bf x}\|_2=1} {\rm QV}({\bf x})\\
{\bf v}_i =\arg \min_{{\bf x}\in W_i^\perp \ {\rm with} \ \|{\bf x}\|_2=1}  {\rm QV}({\bf x})
\end{array}\right.
\end{equation}
inductively for $1\le i\le N-1$, where $\lambda_0=0$,    the initial ${\bf v}_0$ is usually selected by $N^{-1/2} {\bf 1}$,
 and
 $W_i^\perp, 1\le i\le N-1$, are the orthogonal complements of the space spanned by
${\bf v}_j, 0\le j\le i-1$.
We remark that the quadratic variation $\rm QV$ in \eqref{TV.def00} overlook the edge direction in the directed graph setting.
To define GFT on directed graphs, several directed variations to measure the change of signals along   the  graph structure
 have been proposed
 \cite{Sardellitti17, Girault2018, Shafipour19}. 
 The authors in \cite{Sardellitti17, Girault2018, Shafipour19} define frequency and frequency components of GFT on directed graphs via solving
some constrained optimization problems with  directed variations as their objective functions, see  Remark \ref{comparison.remark} for detailed explanation.
 In Section
 \ref{GftLap.section},  we show that right frequency components ${\bf v}_i, 0\le i\le N-1$,
of the proposed GFT  can be obtained via  solving
 constrained minimization problems
 \eqref{CFminimization.eq00} with the objective function ${\rm QV}({\bf x})$
replaced by $ \|{\bf L}{\bf x}\|_2=\sqrt{{\bf x}^T {\bf L}^T{\bf L}{\bf x}}$, see \eqref{frequencycomponents.eq}.
Compared with the GFTs based on constrained optimization of directed variations in
\cite{Sardellitti17} and \cite{Shafipour19},  major differences are that the GFT proposed in this paper coincides with the conventional GFT  \eqref{undirectedLaplacian.def2} in the undirected graph setting, see \eqref{undirectfourier.def}, and that
on directed circulant graphs,  it is essentially the classical discrete Fourier transform, up to certain rotation, permutation and phase adjustment, see Theorem \ref{circulant.thm}.

We say that a graph ${\mathcal E}$ is
 an {\em Eulerian graph} if
the in-degree and out-degree are the same at each vertex,
and that  ${\mathcal E}^T$ is the {\em transpose} of a directed  graph ${\mathcal E}$ if
they have the
same vertex set and  the adjacent matrix  is the transpose of the   adjacent matrix
of the original graph ${\mathcal E}$.
To  measure the ``symmetry" of a directed Eulerian graph ${\mathcal E}$,
in Section \ref{Euleriangraph.section} we consider GFT
${\mathcal F}_t, 0\le t\le 1$,
 on a family of directed graphs  ${\mathcal E}_t, 0\le t\le 1$,
 to connect an Eulerian graph  ${\mathcal E}$ to its transpose  ${\mathcal E}^T$, and
 study  algebraic and analytic properties of the corresponding frequencies and frequency components of the  GFT ${\mathcal F}_t, 0\le t\le 1$, see
 \eqref{Euleriangraph.eq6},
\eqref{sigmalipschitz.eq}, and  Theorems \ref{Euleriangraph.thm}, \ref{Euleriangraph.cor} and \ref{missundirected.thm}.

{\bf Notation}: Bold lower cases and capitals are used to represent the column vectors and matrices respectively.
 Denote the Hermitian and transpose of a matrix $\A$ by $\A^H$ and $\A^T$ respectively,
  and use $\bf 1$, $\bf 0$,  $\bf I$ and $\bf O$ to represent a vector with all 1s,  a row/column vector with all 0s, an identity matrix, and a zero matrix of appropriate size.

\section{Graph Fourier transforms on  directed graphs}\label{freq.sec}

Let $\G=(V, E)$ be
a weighted directed graph of order  $N$
containing no loops or multiple edges, and denote the associated Laplacian  by ${\bf L}={\bf D}-{\bf A}$, where
the adjacent matrix ${\bf A}=(a_{ij})_{i,j\in V}$ has nonzero weights $a_{ij}\ne 0$ only when there is an directed edge from  node $j$ to  node $i$,
and the in-degree matrix ${\bf D}={\rm diag}(d_i)_{i\in V}$ has the in-degree $d_i=\sum_{j\in V}a_{ij}$ of node $i\in V$ as its diagonal entries.
  The Laplacian  ${\mathbf L}$ 
 has eigenvalue zero and the  constant signal ${\bf 1}$ as an associated eigenvector
\begin{equation}\label{Laplacian.def0}
{\bf L}{\bf 1}={\bf 0}, \end{equation}
and in the undirected graph setting, its eigendecomposition
\eqref{undirectedLaplacian.def1}
is used to define GFT  on undirected graphs
\cite{Ricuad2019, stankovic2019introduction,     magoarou2018}.
In this section, 
based on the eigendecomposition \eqref{mathL.decomp}
of the self-adjoint dilation  ${\cal S}(\L)$ of the Laplacian $\L$,
 we propose a novel definition of  GFT and inverse GFT  on  directed graphs, see Definition
 \ref{fourier.def}.
 The proposed GFT   preserves the Parseval identity, see \eqref{parsevalidentity.eq},
 and  in the undirected graph setting,
  it coincides with the conventional GFT in \eqref{undirectedLaplacian.def2}, see \eqref{undirectfourier.def}.
  Circulant graphs have been widely used in image processing
  \cite{ekambaram13}-\cite{ncjs22}.
 In Theorem \ref{circulant.thm}, we  show that  the proposed SVD-based GFT
  on a directed circulant graph is essentially the classical discrete Fourier transform,
  up to certain rotation, permutation and phase adjustment.

Let  orthogonal matrices  $\U=[{\bf u}_0, \ldots, {\bf u}_{N-1}], \V=[{\bf v}_0, \ldots, {\bf v}_{N-1}]$  and  diagonal matrix
 $\pmb \Sigma={\rm diag} (\sigma_0, \ldots, \sigma_{N-1})$
 be as in the  SVD
\eqref{svd.def}
of  the Laplacian  $\L$.
Then 
the self-adjoint dilation  
\begin{equation}\label{mathL.def}
 {\cal S}(\L):=
\begin{pmatrix}
{\bf O}&{\bf L}\\
{\bf L}^T &{\bf O}\end{pmatrix}\in \R^{2N\times 2N }
\end{equation}
 of the Laplacian   $\L$   has the following eigendecomposition,
\begin{equation}\label{mathL.decomp}
{\cal S}(\L)={\bf F} \begin{pmatrix}
\pmb\Sigma&{\bf O}\\
{\bf O} &-\pmb\Sigma\end{pmatrix} {\bf F}^T,
\end{equation}
where
\begin{equation}\label{Fouriermatrix.def}
{\bf F}=\frac{1}{\sqrt 2}\begin{pmatrix}
\U&\U\\
\V&-\V
\end{pmatrix}\in {\mathbb R}^{2N\times 2N}
\end{equation}
is an orthogonal matrix.
 Using the above orthogonal matrix ${\bf F}$, we define the  GFT and inverse  GFT
  on the directed graph ${\mathcal G}$
   as follows.

\begin{definition}\label{fourier.def}
{\rm
Let ${\bf F}$ be the orthogonal matrix  in \eqref{Fouriermatrix.def}.
We define  {\em graph Fourier transform} ${\mathcal F}: {\mathbb R}^N\longmapsto {\mathbb R}^{2N}$
and {\em inverse  graph Fourier transform} ${\mathcal F}^{-1}: {\mathbb R}^{2N}\longmapsto {\mathbb R}^{N}$
on the  directed graph $\G$
 by
 \begin{equation}\label{jointfourier.def}
\mathcal F{\bf x}:={\bf F}^T\begin{pmatrix}
{\bf x}/\sqrt{2}\\
{\bf x}/\sqrt{2}
\end{pmatrix}=\begin{pmatrix}
 ({\bf U}^T+{\bf V}^T){\bf x}/2\\
 ({\bf U}^T-{\bf V}^T){\bf x}/2
\end{pmatrix}
\end{equation}
and
\begin{equation}\label{inverFT.def}
{\mathcal F}^{-1}\begin{pmatrix}
{\bf z}_1\\
{\bf z}_2
\end{pmatrix}\hskip-0.03in:=\hskip-0.05in \begin{pmatrix}
\frac{{\bf I}}{\sqrt{2}} \ \frac{{\bf I}}{\sqrt{2}}
\end{pmatrix}{\bf F} \begin{pmatrix}
{\bf z}_1\\
{\bf z}_2
\end{pmatrix}\hskip-0.03in=\hskip-0.03in \frac{1}{2} \big({\bf U} ({\bf z}_1+{\bf z}_2)+ {\bf V}({\bf z}_1-{\bf z}_2)\big),
\qquad  
\end{equation}
 where ${\bf x}$ is a graph signal on $\G$ and ${\bf z}_1, {\bf z}_2\in \R^N$. 
}\end{definition}

The GFT in \eqref{jointfourier.def}  provides a tool 
to analyze and represent signals in the spectral domain.
Shown  in  Figure \ref{piecewise.fig} is a piecewise constant signal on
      a weighted Minnesota traffic graph (left),
      the frequencies $\sigma_i, 0\le i\le N-1$, of the proposed GFT (second),
       and the first and next $N$-th components of the  GFT of
      a piecewise  constant signal  on the graph (third and right).
We observe that the piecewise constant signal on the weighted Minnesota graph 
has 
its energy concentrated mainly at  low frequencies.

\begin{figure}[t]
\includegraphics[width=42mm, height=36mm]{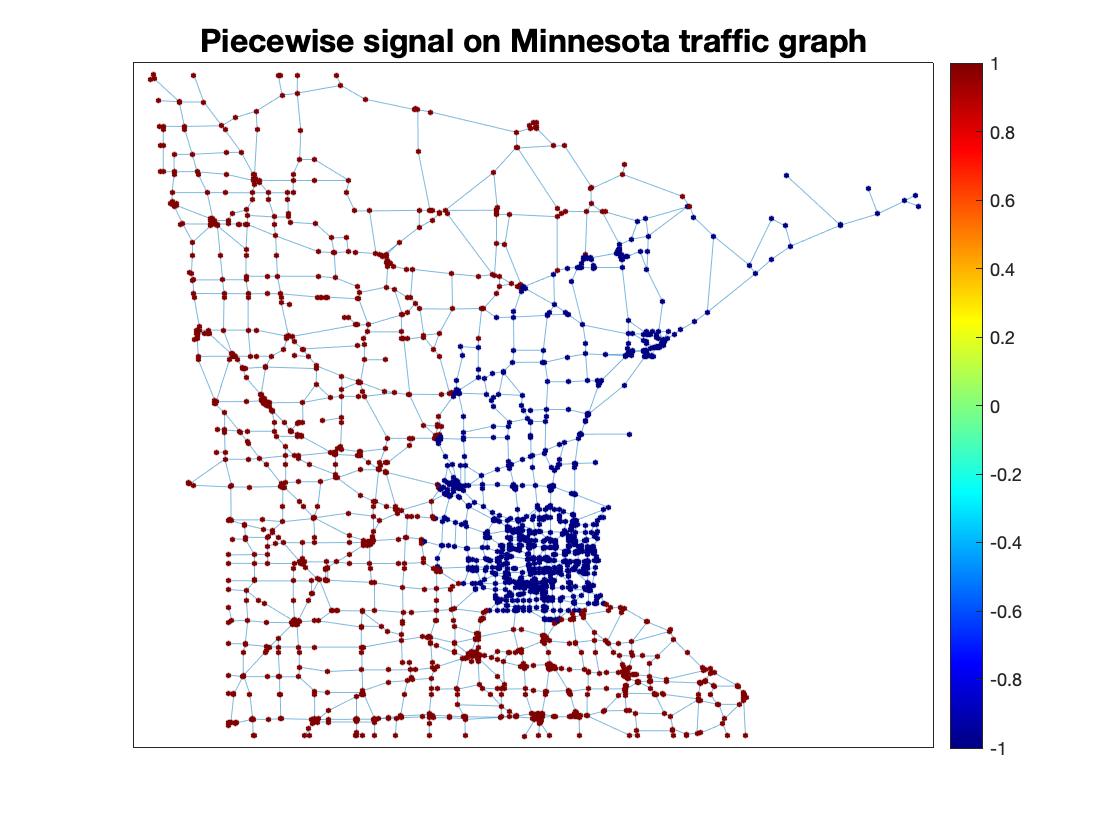}
\includegraphics[width=42mm, height=36mm]{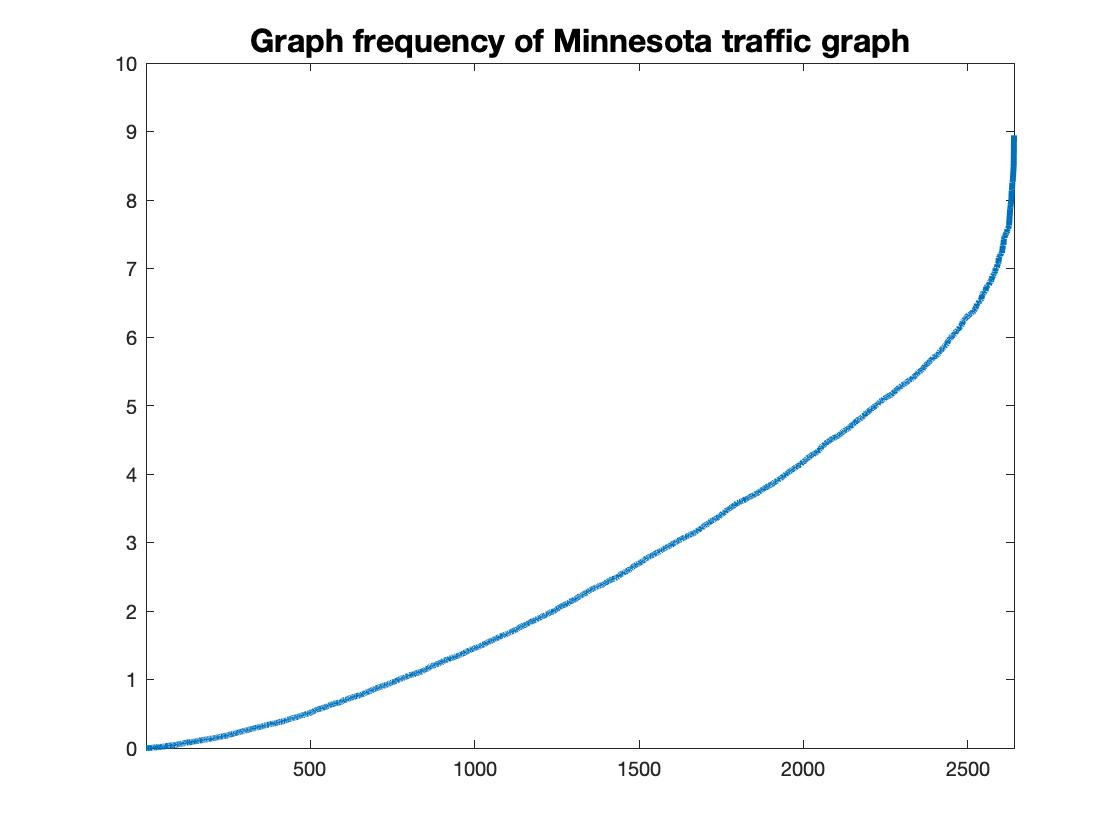}
\includegraphics[width=42mm, height=36mm]{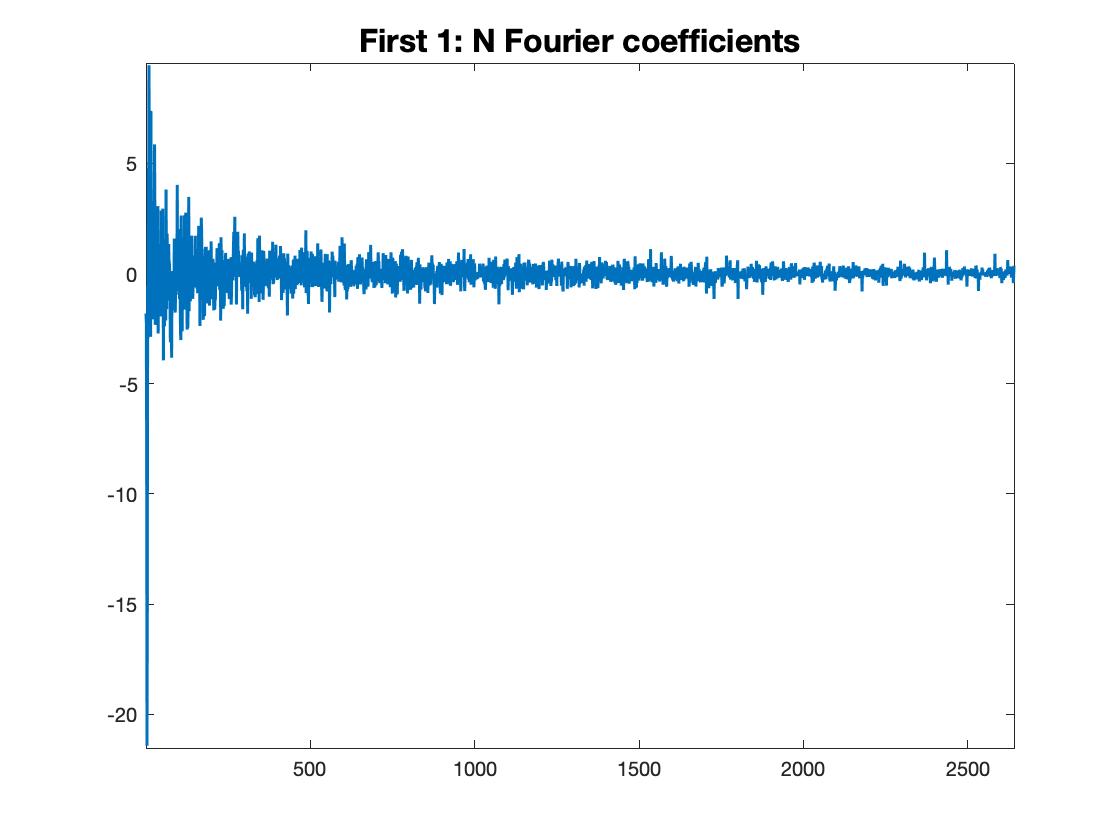}
\includegraphics[width=42mm, height=36mm]{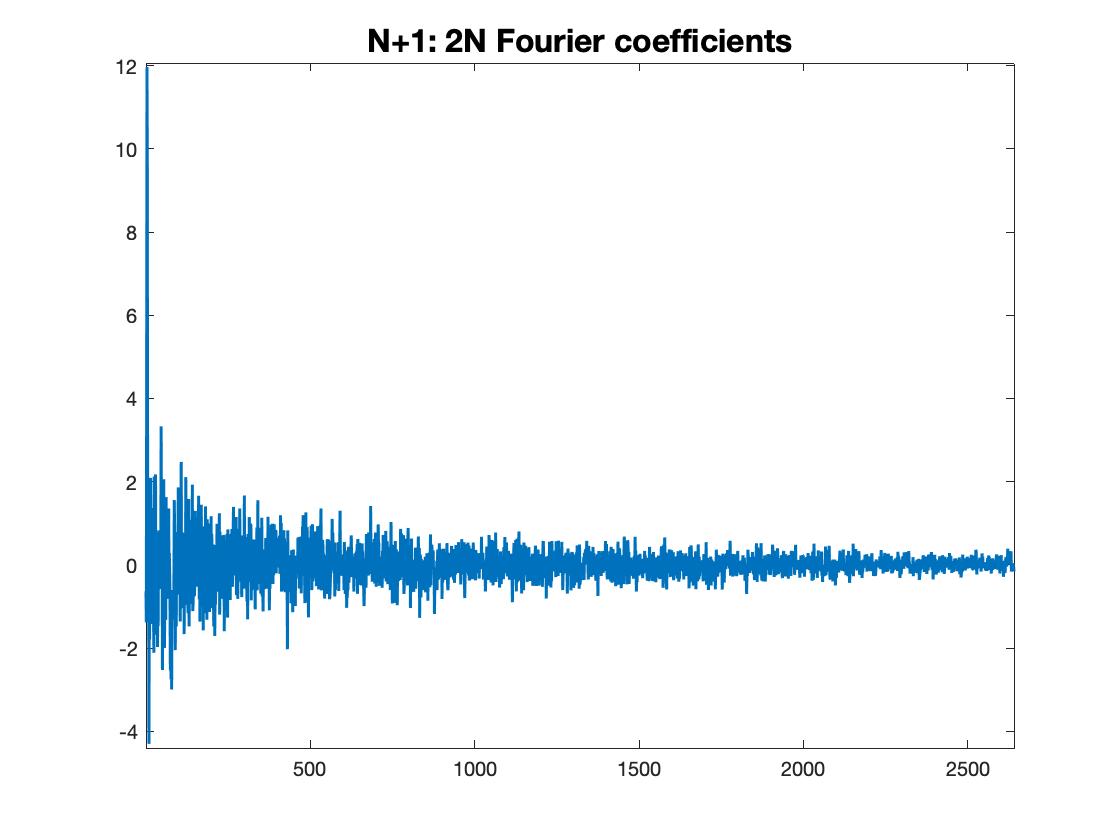}
\caption{Plotted on the  left and the second are
 a  piecewise constant signal ${\bf x}_0$ on a weighted  Minnesota traffic graph of order $N=2640$ with
the weights $w_{ij}$ on adjacent edges  $(j,i)$  being randomly chosen in the interval $[0, 2]$, and
 the frequencies $\sigma_i, 0\le i\le  N-1$, in \eqref{frequency.def} of our SVD-based GFT respectively.
On the third and fourth are the first $N$-component
 $({\bf U}^T+{\bf V}^T){\bf x}_0/2$ and the next $N$-components $({\bf U}^T-{\bf V}^T){\bf x}_0/2$
of the GFT ${\mathcal F}{\bf x}_0$ of the  signal ${\bf x}_0$ plotted on the left. The relative percentage of signal energy
$(\sum_{i=0}^{M-1} |({\bf u}_i+{\bf v}_i)^T {\bf x}_0/2|^2+ |({\bf u}_i-{\bf v}_i)^T {\bf x}_0/2|^2)^{1/2}/\|{\bf x}_0\|_2$
for the first $M= 20,  50$ and $100$ frequencies are $0.6964, 0.7602, 0.8063$. 
 }
\label{piecewise.fig}
\end{figure}

By the orthogonality of the matrix  ${\bf F}$, one may verify that
the Parseval's identity  holds, 
 \begin{equation}\label{parsevalidentity.eq}
 \|{\bf x}\|_2=\|\mathcal F{\bf x}\|_2, \ {\bf x}\in\R^N,
 \end{equation}
and the inverse GFT  ${\mathcal F}^{-1}$ is the pseudo-inverse of the GFT ${\mathcal F}$, i.e.,
\begin{equation}
{\mathcal F}^{-1} \begin{pmatrix}
{\bf z}_1\\
{\bf z}_2
\end{pmatrix}={\rm arg}\min_{{\bf z}\in {\mathbb R}^N} \left\|{\mathcal F}{\bf z}-\begin{pmatrix}
{\bf z}_1\\
{\bf z}_2
\end{pmatrix}\right\|_2, \ {\bf z}_1, {\bf z}_2\in {\mathbb R}^N.
\end{equation}
Therefore the original graph signal ${\bf x}$ can be reconstructed from its GFT ${\mathcal F}{\bf x}$,
\begin{equation}
{\mathcal F}^{-1} {\mathcal F} {\bf x}={\bf x}, \ \ {\bf x}\in {\mathbb R}^N.
\end{equation}
%


For the case that the graph ${\mathcal G}$ is undirected, 
orthogonal matrices  ${\bf U}$ and ${\bf V}$ in  \eqref{svd.def} can be selected to be the same, i.e.,
${\bf U}={\bf V}$. Then the corresponding GFT ${\mathcal F}{\bf x}$
 of a graph signal ${\bf x}$ becomes
\begin{equation}\label{undirectfourier.def}
{\mathcal F}{\bf x}= \begin{pmatrix}
{\bf V}^T{\bf x}\\
{\bf 0}
\end{pmatrix}.
\end{equation}
This shows that,  in the undirected graph setting, the proposed SVD-based GFT is essentially the same as  the well-accepted GFT  \eqref{undirectedLaplacian.def2}
on undirected graphs.

\smallskip

For $N\ge 1$ and a set $Q=\{q_1, \ldots, q_L\}$ of positive integers ordered with $1\le q_1<\ldots< q_L\le N-1$, let
the  {\em directed  circulant graph} ${\mathcal C}_d:={\mathcal C}_d(N, Q)$  generated by  $Q$
be the unweighted graph with  the vertex set   $V_N=\{0, 1, \ldots, N-1\}$  and the edge set
$E_N(Q)=\{(i, i+q\ {\rm mod}\ N),\  i\in V_N, q\in Q\}$,  
where  we say that
 $a=b\ {\rm mod }\ N$ if $(a-b)/N$ is an integer
  \cite{ ekambaram13}-\cite{ncjs22}.
Set \begin{equation}  \label{circulant.symbol}
 P(z)=L-\sum_{l=1}^L z^{q_l}, \end{equation}
 and
denote the  discrete Fourier transform  matrix by
\begin{equation}
\label{wdft.def}
{\bf W}:= \big(N^{-1/2} \omega_N^{ij}\big)_{0\le i, j\le N-1},\end{equation}
where  $\omega_N=\exp (2\pi \sqrt{-1}/N)$ is the $N$-th root of the unit.
   One may verify that   the Laplacian   $\L_{{\mathcal C}_d}$
  on the directed circulant graph ${\mathcal C}_d$ is a circulant matrix
  that has
 eigenvalues
  \begin{equation}\label{symbolvalue.decomp}
  P(\omega_N^{i})=|P(\omega_N^{i})|\exp(\sqrt{-1}\theta_i), 0\le i\le N-1,
  \end{equation}
 and
the $i$-th column of the  discrete Fourier transform  matrix
${\bf W}$ as an unit eigenvector associated with the eigenvalue
 $ P(\omega_N^{i}), 0\le i\le N-1$,
  where
 \begin{equation*}\exp(\sqrt{-1}\theta_i)=\left\{\begin{array}
{ll} 1 & {\rm if} \ P(\omega_N^{i})=0 \\
P(\omega_N^{i})/|P(\omega_N^{i})|  & {\rm if} \  P(\omega_N^{i})\ne 0.
\end{array}\right.\end{equation*}

Let
\begin{equation} \label{circulant.thm.eq1}
{\bf R}=\left\{\begin{array}
{ll} {\rm diag} (1, {\bf R}_2, \ldots, {\bf R}_2)  & {\rm if} \ N \ {\rm is \ odd} \\
{\rm diag} (1, {\bf R}_2, \ldots, {\bf R}_2, 1)  & {\rm if} \ N \ {\rm is \ even}
\end{array}\right.
\end{equation}
be the block diagonal matrix with number one and the $2\times 2$ unitary matrix
${\bf R}_2=\frac{1}{\sqrt{2}}\begin{pmatrix}
1&-\sqrt{-1}\\
1&\sqrt{-1}
\end{pmatrix}$ as its diagonal blocks, and
let the  diagonal matrix
\begin{equation} \label{circulant.thm.eq2}
{\bf \Theta}=
{\rm diag} \big(\exp(\sqrt{-1}\theta_0), \ldots, \exp(\sqrt{-1}\theta_{N-1})\big) \end{equation}
have  phases $\exp(\sqrt{-1}\theta_i)$  in \eqref{symbolvalue.decomp}
as its diagonal entries.
In Proposition \ref{singular.prop} of Appendix \ref{TheoremCirculant.proof}, we show that the Laplacian matrix  $\L_{{\mathcal C}_d}$
  on the directed circulant graph ${\mathcal C}_d$ has the following SVD,
   \begin{equation}
  \label{circulant.thm.eq3}
  \L_{{\mathcal C}_d}= {\bf U} {\pmb \Sigma} {\bf V}^T,
 \end{equation}
 where  ${\bf P}_0$ and ${\bf P}_1$ are permutation  matrices (see \eqref{circulant.thm.pfeq5} and \eqref{circulant.thm.pfeq6} for explicit expressions),
 \begin{equation} \label{circulant.thm.eq4}
 {\bf U}={\bf W}{\bf \Theta} {\bf P}_0 {\bf R} {\bf P}_1\ {\rm and} \
 {\bf V}= {\bf W} {\bf P}_0 {\bf R} {\bf P}_1
 \end{equation}
are orthogonal matrices with real entries, and
 \begin{equation} \label{circulant.thm.eq5}
 {\pmb \Sigma}={\rm diag}(\sigma_0, \ldots, \sigma_{N-1})\end{equation}
has diagonal entries being nondecreasing rearrangement
of the magnitudes
   $|P(\omega_N^{i})|, 0\le i\le N-1$,  in \eqref{symbolvalue.decomp}.
Based on the above  SVD of
 the Laplacian matrix  $\L_{{\mathcal C}_d}$, we observe that
 the  GFT
  in Definition \ref{fourier.def}
 is  essentially
 the classical discrete Fourier transform,
 \begin{equation} \label{dft.def} {\rm DFT}({\bf x}):={\bf W}^{H}{\bf x}, \  {\bf x}\in {\mathbb R}^N,
\end{equation}
  up to certain rotation  ${\bf R}$, phase adjustment ${\bf \Theta}$ and permutations  ${\bf P}_0$ and ${\bf P}_1$.

 \begin{theorem}\label{circulant.thm} {\rm
Let  $N\ge 1, Q=\{q_1, \ldots, q_L\}$ be a set of positive integers  with $1\le q_1<\ldots<q_L\le N-1$,
${\mathcal C}_d(N, Q)$ be the directed circulant graph generated by  $Q$,
and take the SVD
  \eqref{circulant.thm.eq3}  of
   the Laplacian matrix  $\L_{{\mathcal C}_d}$
  on  ${\mathcal C}_d(N, Q)$.
  Then the corresponding
GFT  in \eqref{jointfourier.def} is given by
  \begin{eqnarray*}\label{circulant.thm.eq6}
{\mathcal F}{\bf x}  & \hskip-0.08in = & \hskip-0.08in  \frac{1}{2}\begin{pmatrix}
{\bf P}_1   &{\bf O}\\
{\bf O} & {\bf P}_1
\end{pmatrix}
\begin{pmatrix}
 {\bf R}  &{\bf O}\\
{\bf O} &  {\bf R}
\end{pmatrix}^H
\begin{pmatrix}
 {\bf P}_0  &{\bf O}\\
{\bf O} &  {\bf P}_0
\end{pmatrix}\begin{pmatrix}
{\bf \Theta}&{\bf \Theta}\\
{\bf I} &- {\bf I}
\end{pmatrix}^H \begin{pmatrix}
{\rm DFT}({\bf x})\\
{\rm DFT}({\bf x})
\end{pmatrix},\end{eqnarray*}
where ${\bf x}$ is a  signal on  ${\mathcal C}_d(N, Q)$, and the rotation ${\bf R}$, the phase adjustment matrix $\bf \Theta$,  and the permutations ${\bf P}_0$ and ${\bf P}_1$, are
given in \eqref{circulant.thm.eq1}, \eqref{circulant.thm.eq2},
 \eqref{circulant.thm.pfeq5} and  \eqref{circulant.thm.pfeq6}  respectively.
}
\end{theorem}

The proof is shown in Appendix \ref{TheoremCirculant.proof}.

\section{Graph Fourier transform and graph Laplacian}
\label{GftLap.section}

Let $\G=(V, E)$ be
a directed graph of order  $N$
containing no loops or multiple edges, and
$\U=[{\bf u}_0, \ldots, {\bf u}_{N-1}]$, $\V=[{\bf v}_0, \ldots, {\bf v}_{N-1}]$
and $\pmb \Sigma={\rm diag} (\sigma_0, \ldots, \sigma_{N-1})$
be the orthogonal matrices  and  diagonal matrix
in the SVD \eqref{svd.def}
of   the associated Laplacian  ${\bf L}$.
In this paper, we propose to use singular values  $\sigma_i, 0\le i\le N-1$, of Laplacian ${\mathbf L}$ to carry the
graph frequencies of the GFT in Definition \ref{fourier.def}, and to take
the columns ${\bf u}_i$ and ${\bf v}_i, 0\le i\le N-1$, of orthogonal matrices ${\bf U}$ and ${\bf V}$ as
its left/right frequency components.
We observe that  frequencies  of the  proposed   GFT
have similar pattern to the ones  in
\cite{Singh16, Sardellitti17}, see Figure \ref{directedcyclefrequencies.fig}.
  Based on  the SVD \eqref{svd.def},  we propose an effective algorithm \eqref{frequencycomponents.eq}
   to evaluate frequencies and
   left/right frequency components  and hence the GFT of a graph signal.
 In Remark \ref{comparison.remark} of this section, we  compare the proposed  SVD-based  GFT with the GFTs
 in
 \cite{Sardellitti17, Shafipour19}
to be defined by solving some constrained optimization problems.

Let ${\cal S}(\L)$  be the self-adjoint dilation \eqref{mathL.decomp} of the Laplacian matrix  $\L$ on the directed graph ${\mathcal G}$,
and ${\bf F}$ be  as in
\eqref{Fouriermatrix.def}.
By \eqref{mathL.decomp} and  \eqref{Fouriermatrix.def}, we have
\begin{equation}\label{comm.eq.1}
{\bf  F}^T{\cal  S}(\L)=\begin{pmatrix}
{\pmb\Sigma}&{\bf O}\\
{\bf O} &-{\pmb\Sigma}\end{pmatrix}{\bf F}^T,
\end{equation}
where the diagonal matrix $\pmb\Sigma={\rm diag}(\sigma_0, \ldots, \sigma_{N-1})$ has  singular values of the Laplacian ${\bf L}$
as its diagonal entries.
Recall that  in the undirected graph setting,  the SVD \eqref{svd.def} of the Laplacian ${\mathbf L}$ is the same as
its eigendecomposition \eqref{undirectedLaplacian.def1} and the
proposed  SVD-based GFT is essentially the GFT \eqref{undirectedLaplacian.def2} by \eqref{undirectfourier.def}.
So following the terminology of GFT on undirected graphs,
%
we use
 \begin{equation}\label{frequency.def}
    0\le\sigma_0\le\sigma_1\le\ldots\le\sigma_{N-1}
    \end{equation}
    to carry the notion of   {\em  frequencies} of the  GFT in Definition \ref{fourier.def}, see Figure \ref{piecewise.fig} for frequencies of a
    weighted Minnesota traffic graph of size
$N=2640$.

  By \eqref{Laplacian.def0},
 the  GFT   in Definition \ref{fourier.def} has zero as its lowest frequency, i.e.,
    \begin{equation}
    \sigma_0=0.
    \end{equation}
Shown in Figure \ref{directedcyclefrequencies.fig} are a  directed unweighted graph of size $15$ containing three clusters connected with a directed cycle \cite[Fig. 1(c)]{Sardellitti17}, and
its  frequencies
obtained by  the splitting orthogonality constraint method (SOC)  \cite[Algorithm 1]{Sardellitti17},
the proximal alternating minimized augmented Lagrangian  methods (PAMAL)  \cite[Algorithms 2 and 3]{Sardellitti17}, the Jordan decomposition method (Jordan)
in  \eqref{Jordan.gft.def}  \cite{sandryhaila2013, sandryhaila14,  deri2017, Domingos20},   and the  SVD-based approach proposed in this paper.
It is observed that frequencies  \eqref{Fouriermatrix.def} obtained by our  approach have similar pattern to the ones in  SOC, PAMAL and Jordan. 

\begin{figure}[t]
\begin{center}
\hskip .05in \includegraphics[width=62mm, height=46mm]{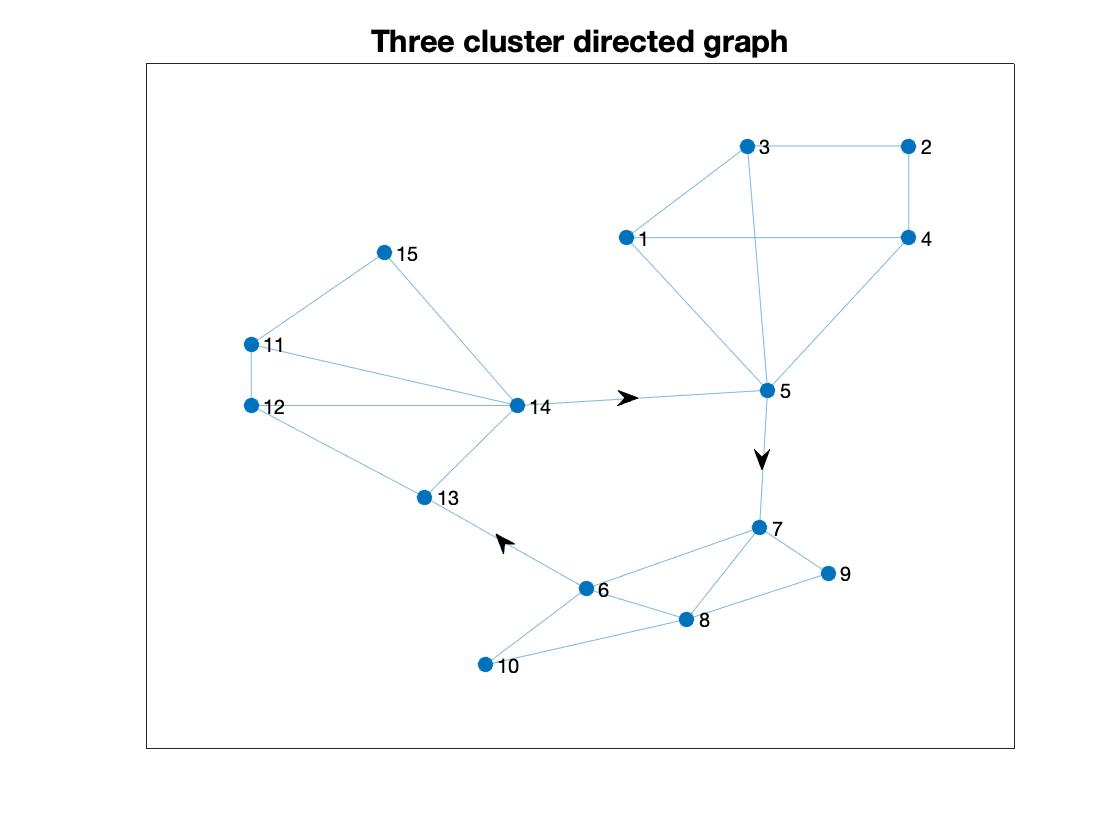} \hskip .05in
\includegraphics[width=62mm, height=46mm]{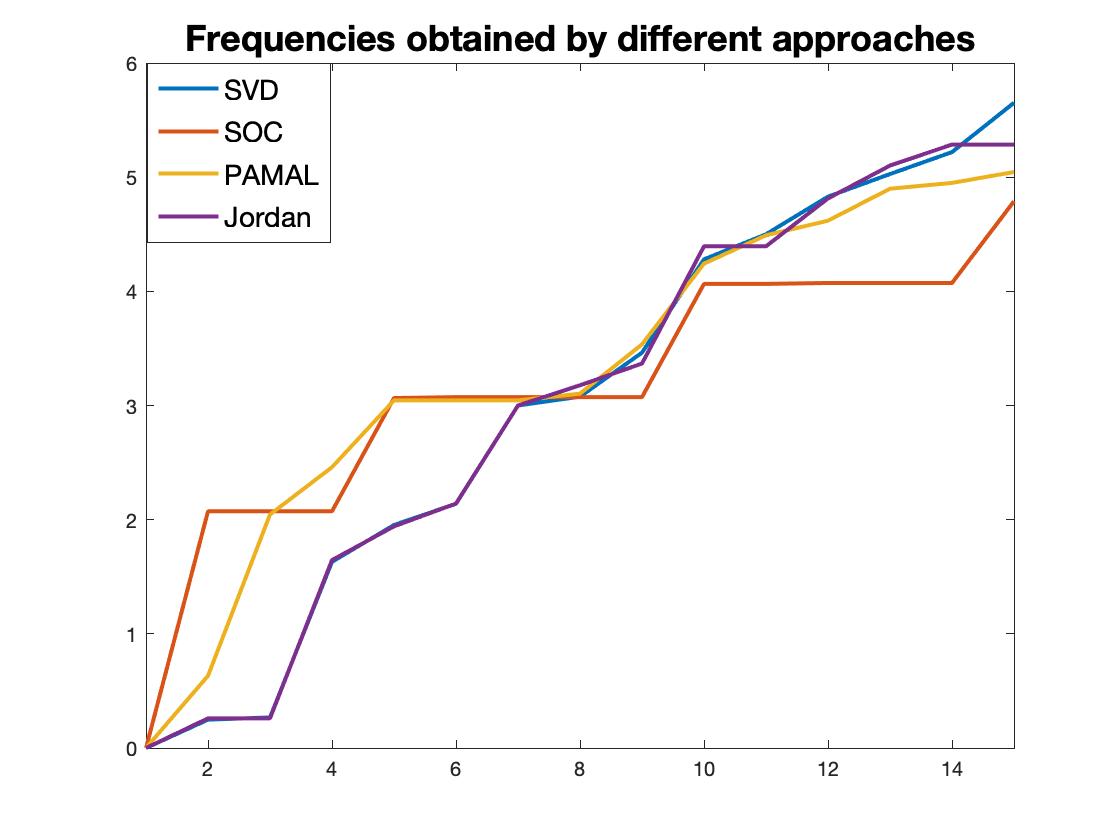}\\
\includegraphics[width=31mm, height=28mm]{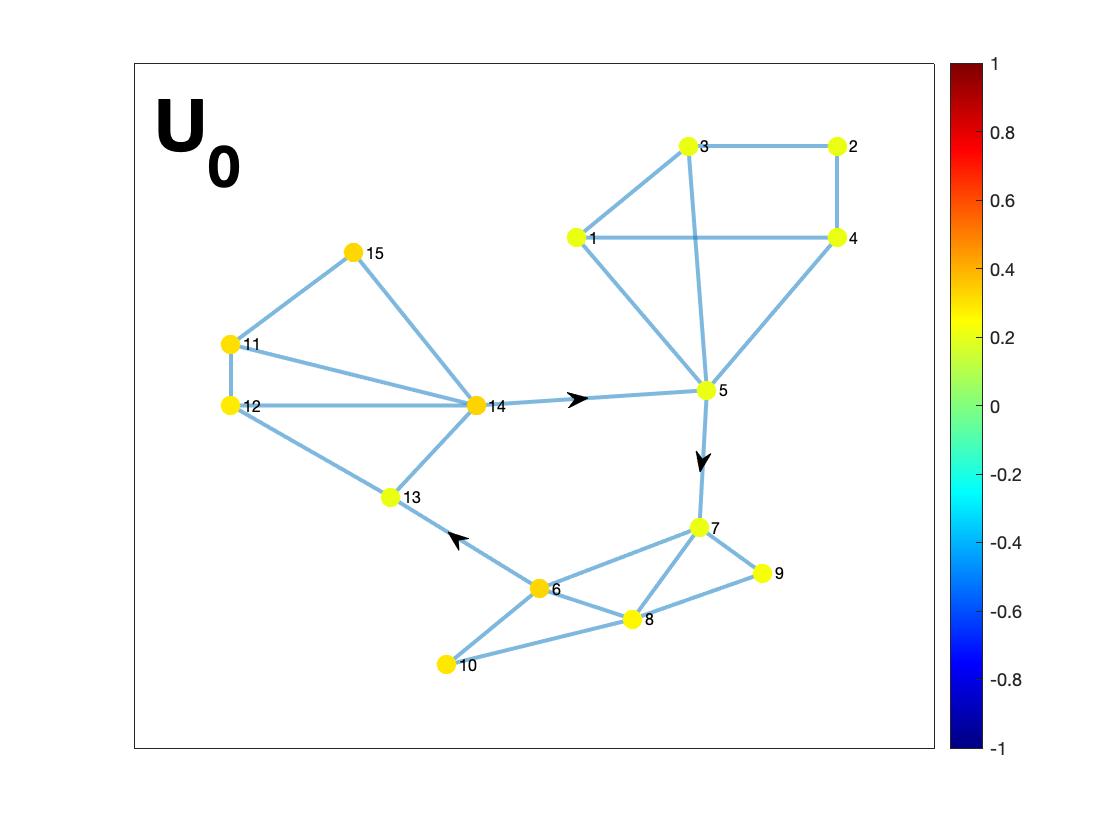}
\includegraphics[width=31mm, height=28mm]{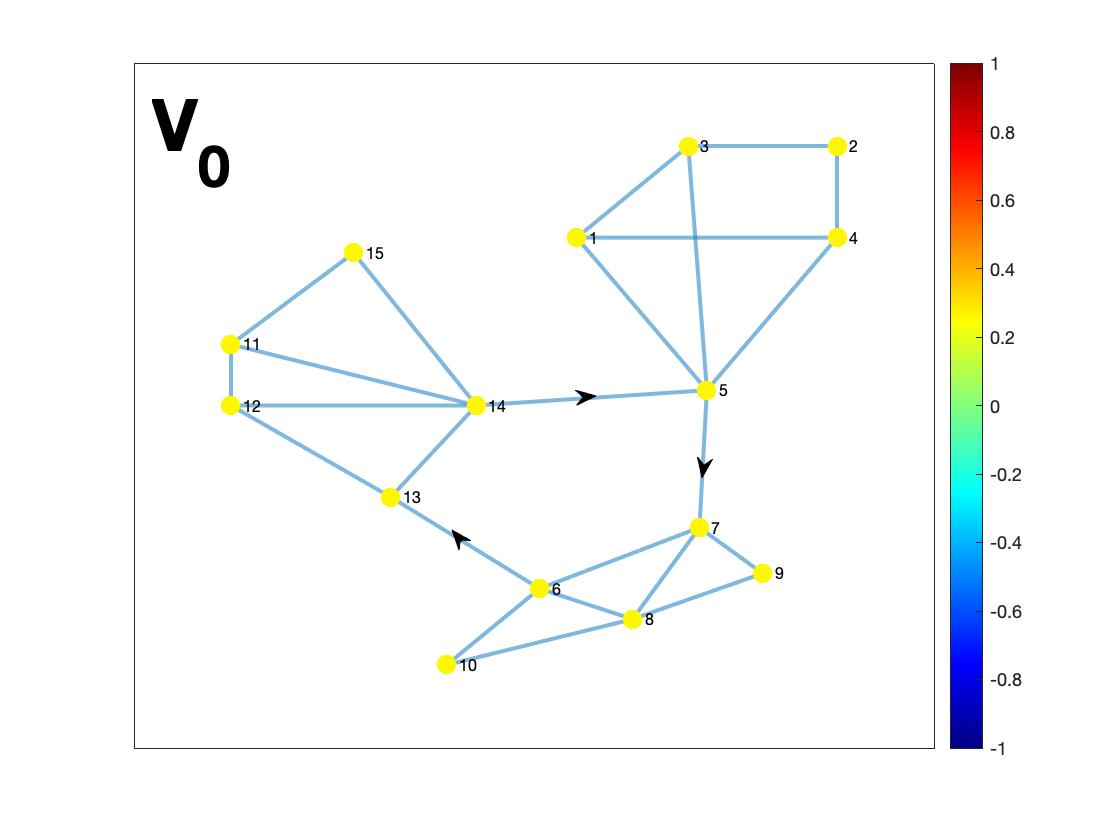}
\includegraphics[width=31mm, height=28mm]{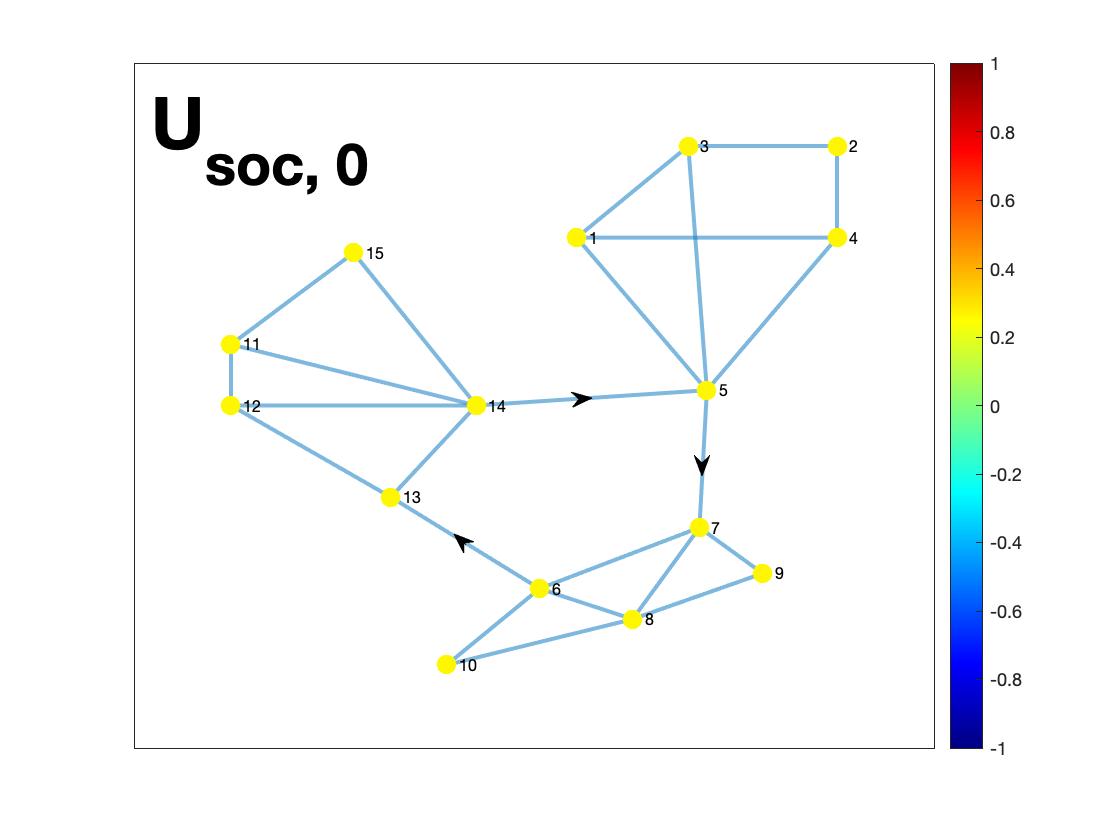}
\includegraphics[width=31mm, height=28mm]{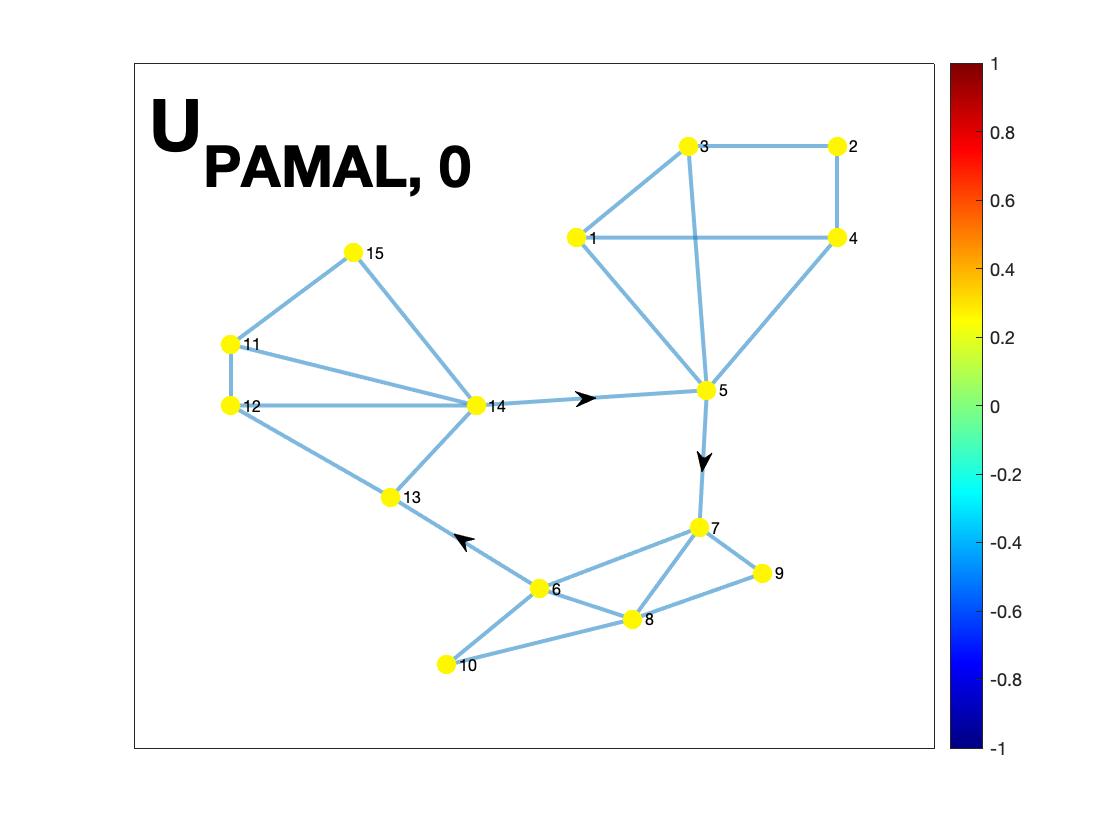}\\
\includegraphics[width=31mm, height=28mm]{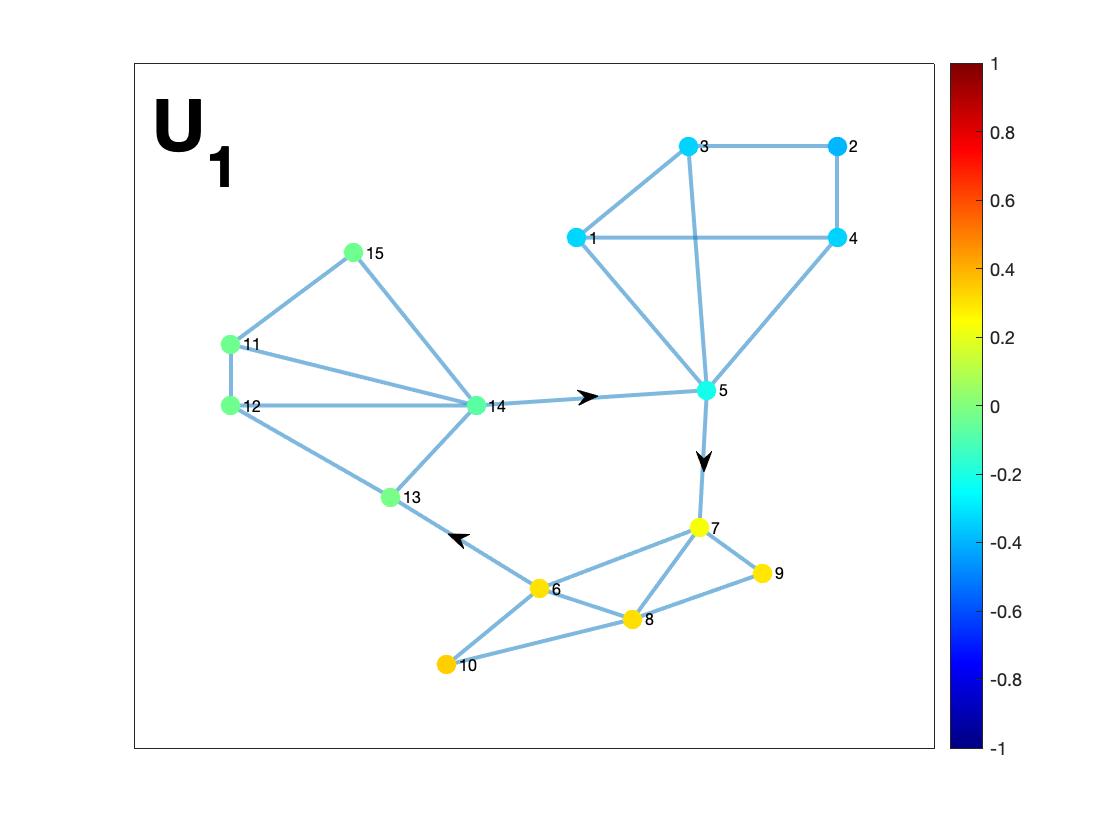}
\includegraphics[width=31mm, height=28mm]{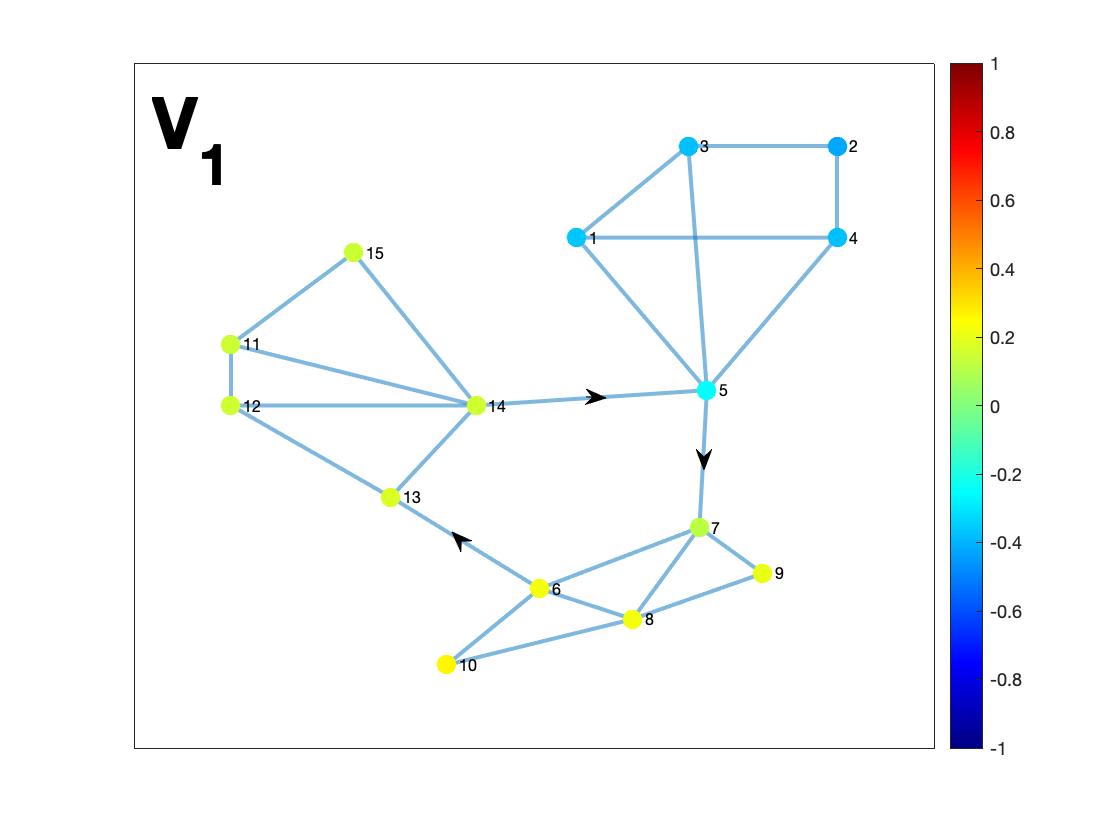}
\includegraphics[width=31mm, height=28mm]{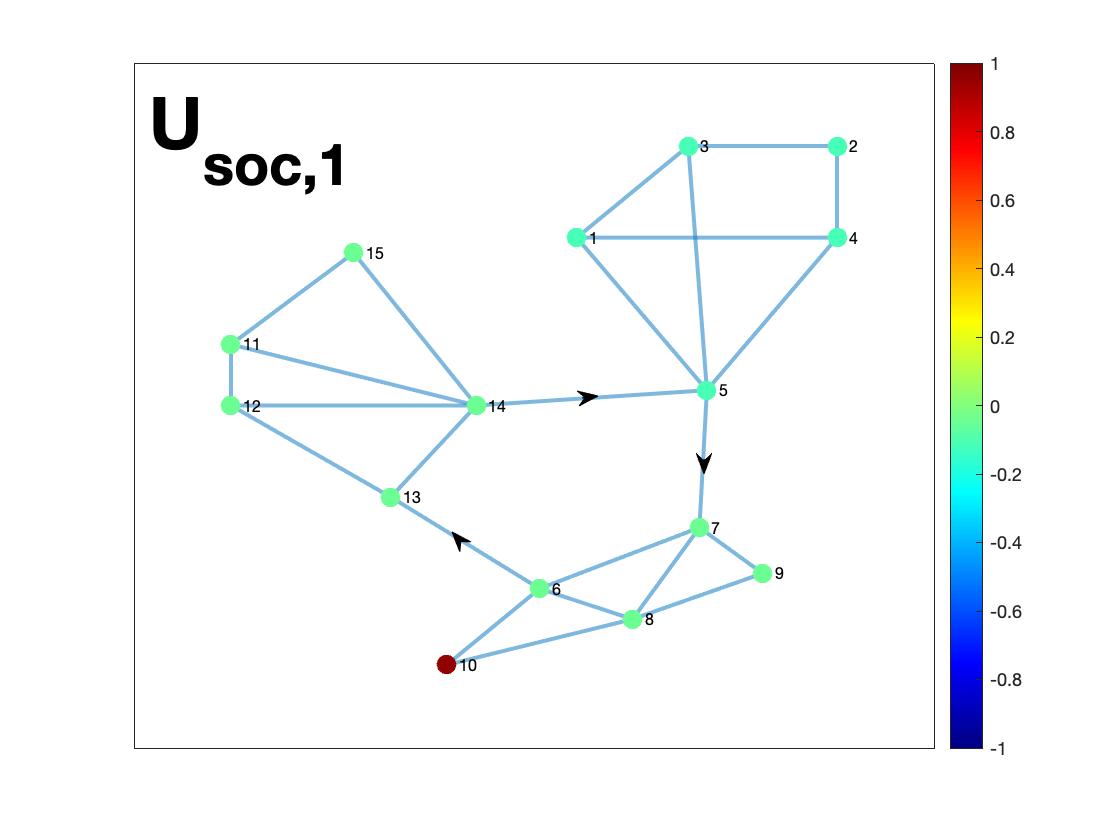}
\includegraphics[width=31mm, height=28mm]{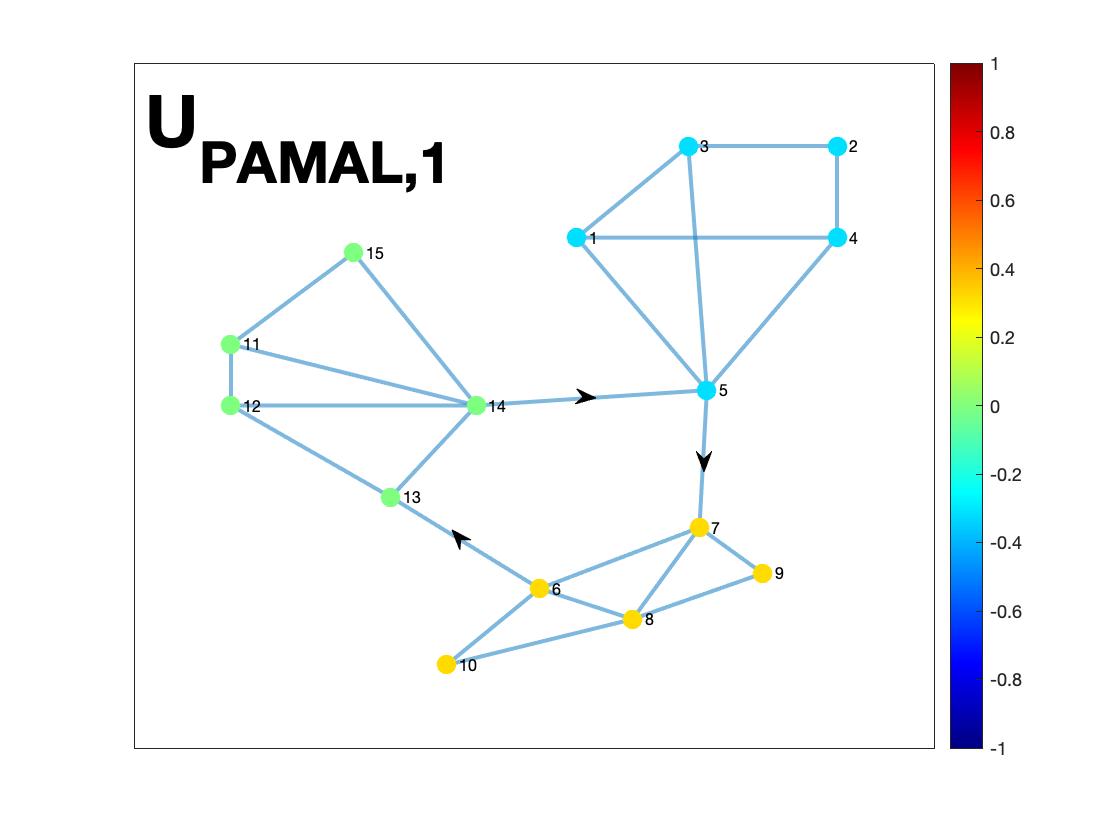}\\
\includegraphics[width=31mm, height=28mm]{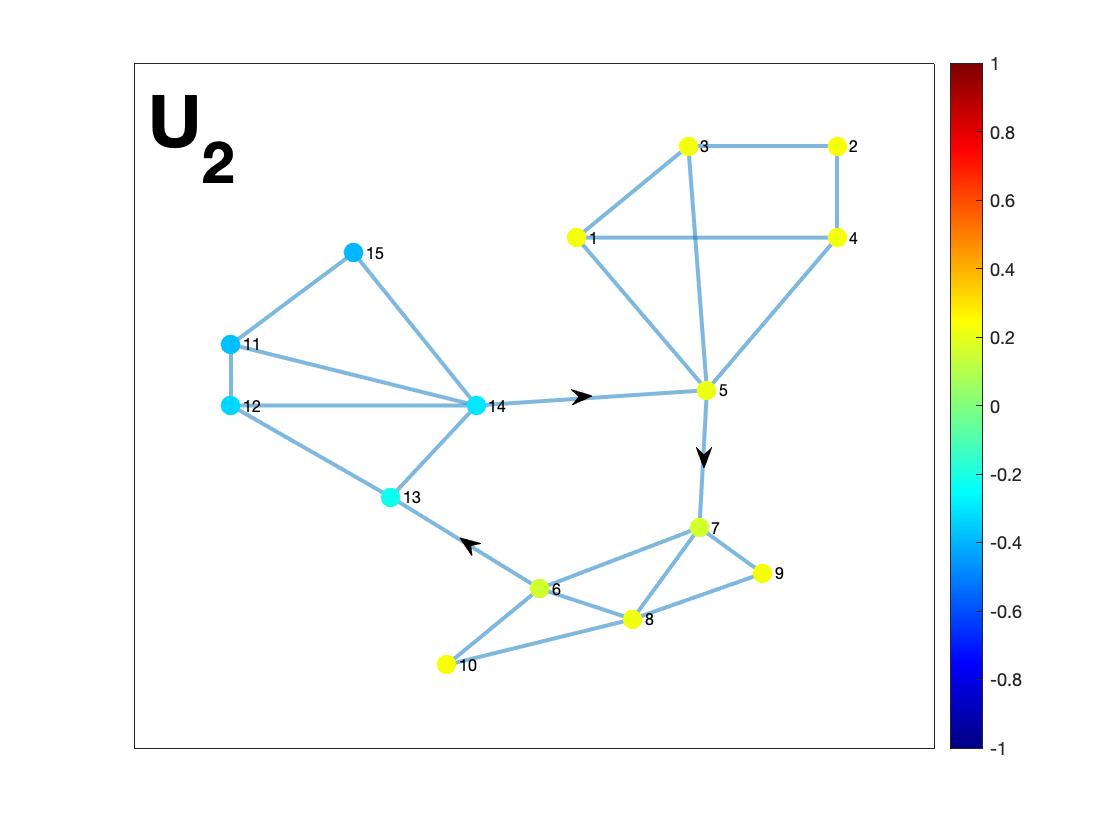}
\includegraphics[width=31mm, height=28mm]{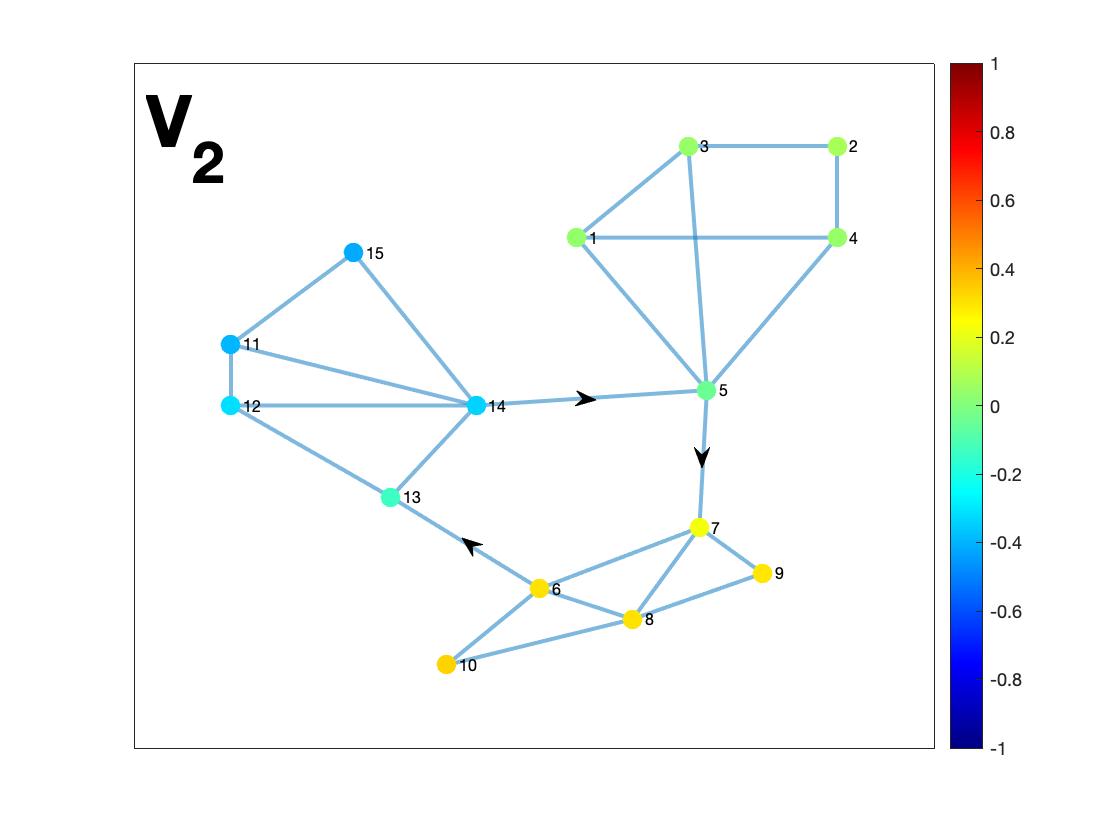}
\includegraphics[width=31mm, height=28mm]{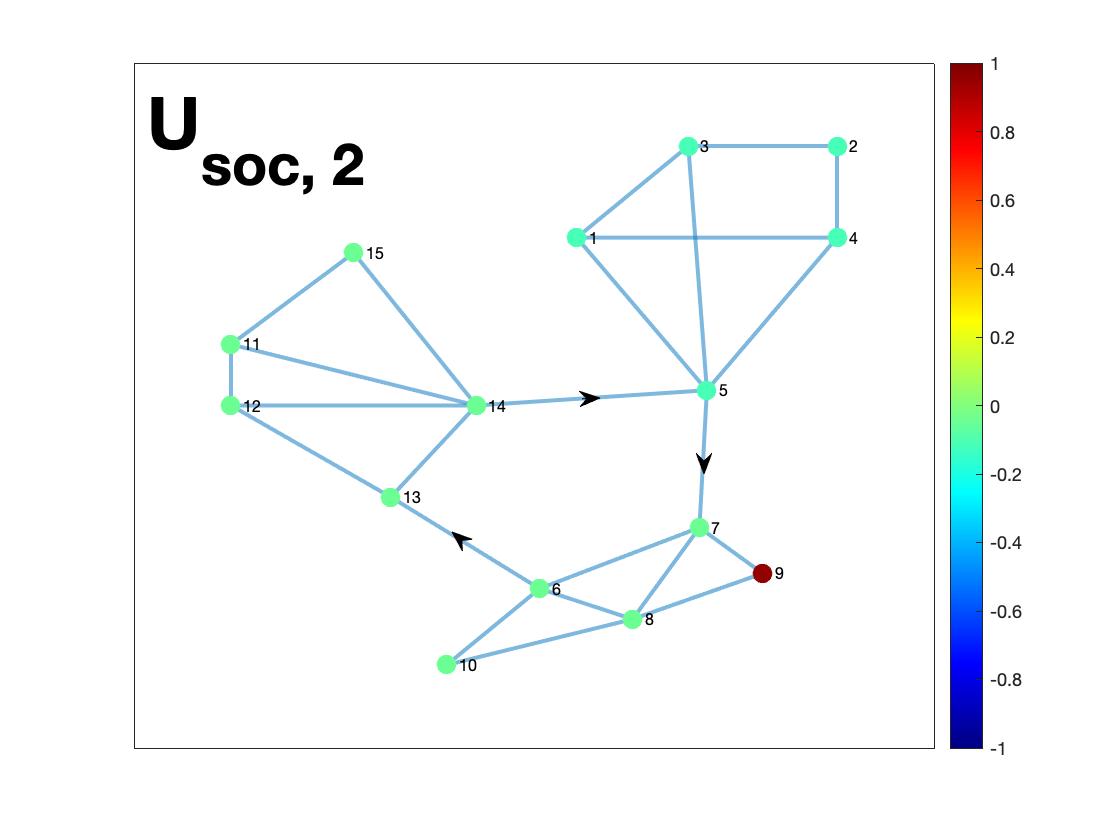}
\includegraphics[width=31mm, height=28mm]{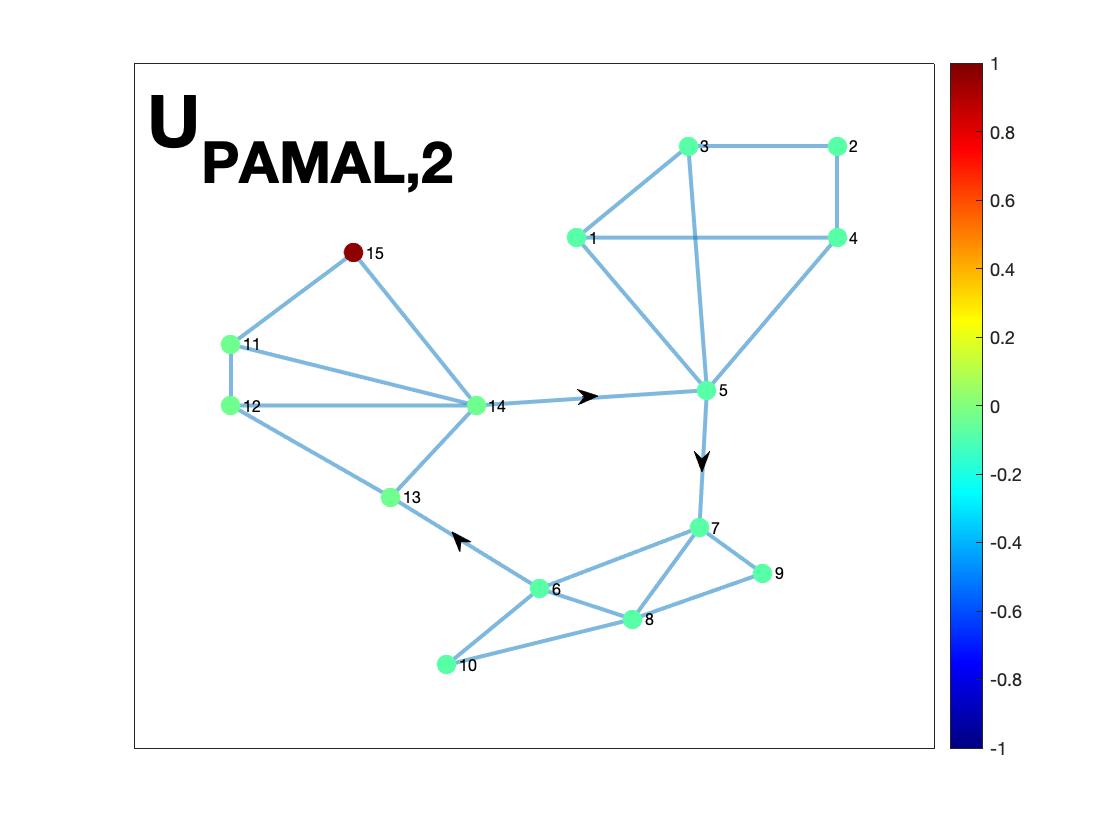}\\
\includegraphics[width=31mm, height=28mm]{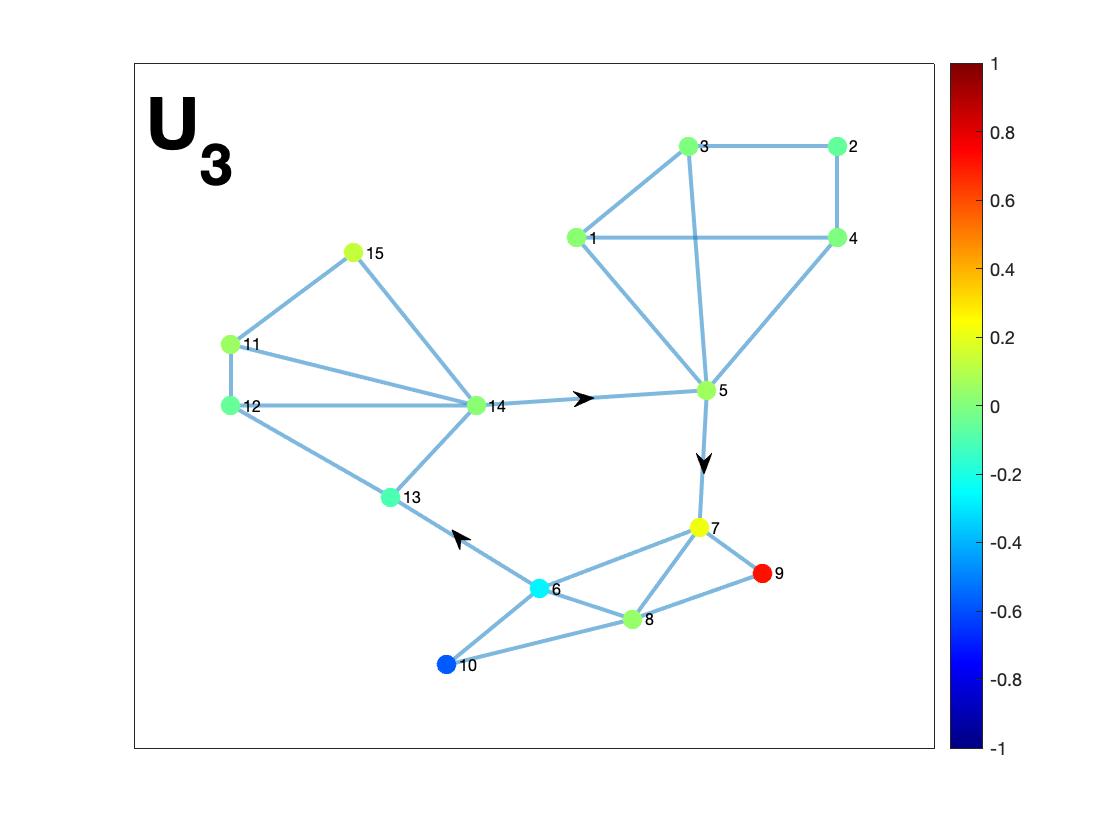}
\includegraphics[width=31mm, height=28mm]{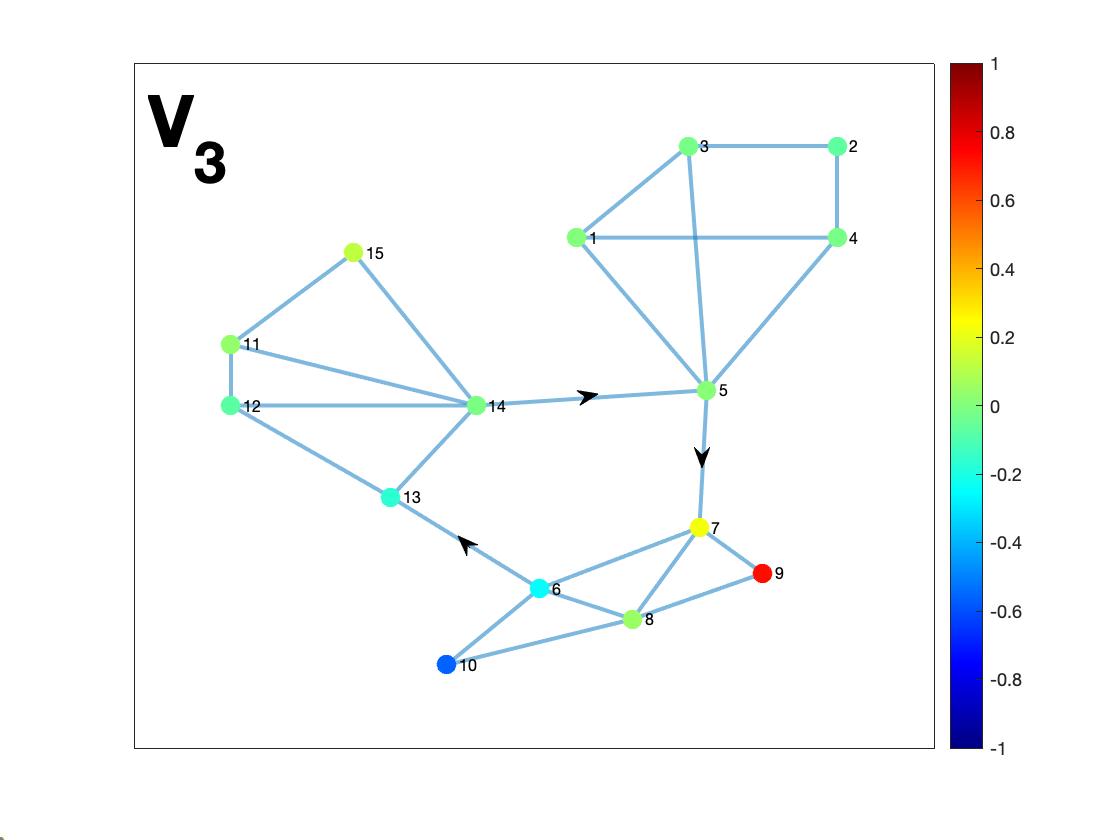}
\includegraphics[width=31mm, height=28mm]{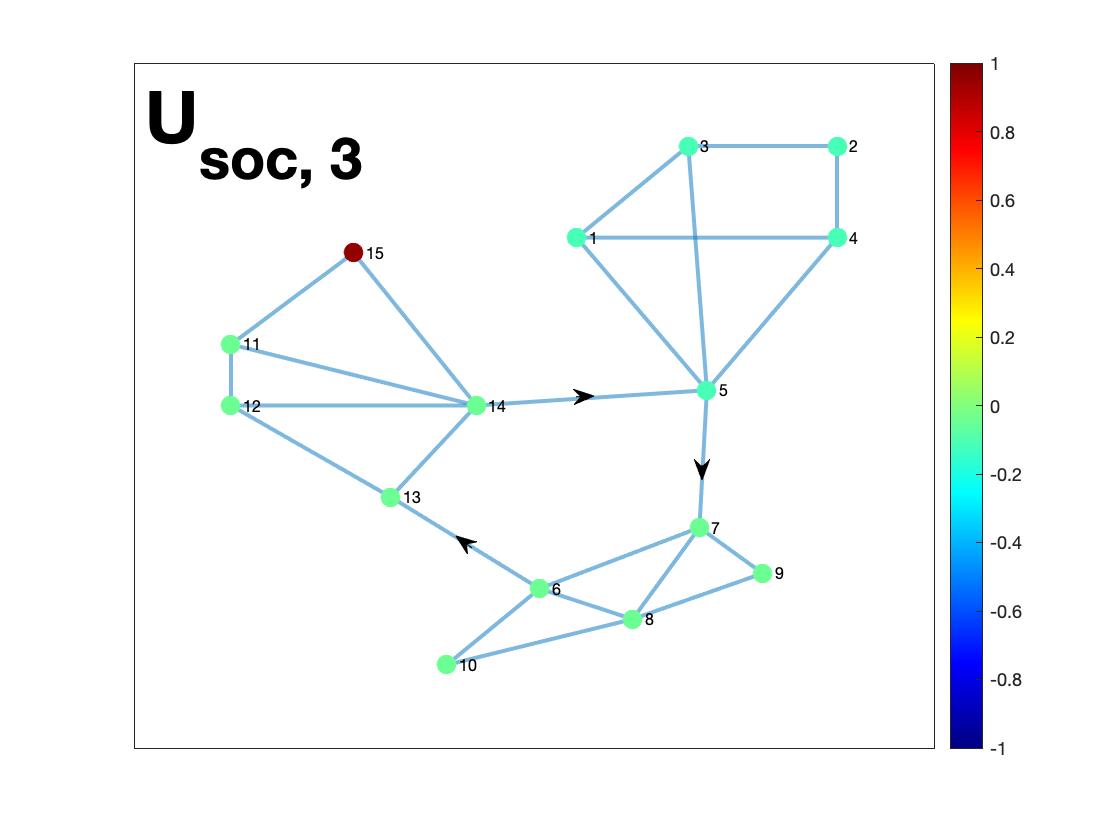}
\includegraphics[width=31mm, height=28mm]{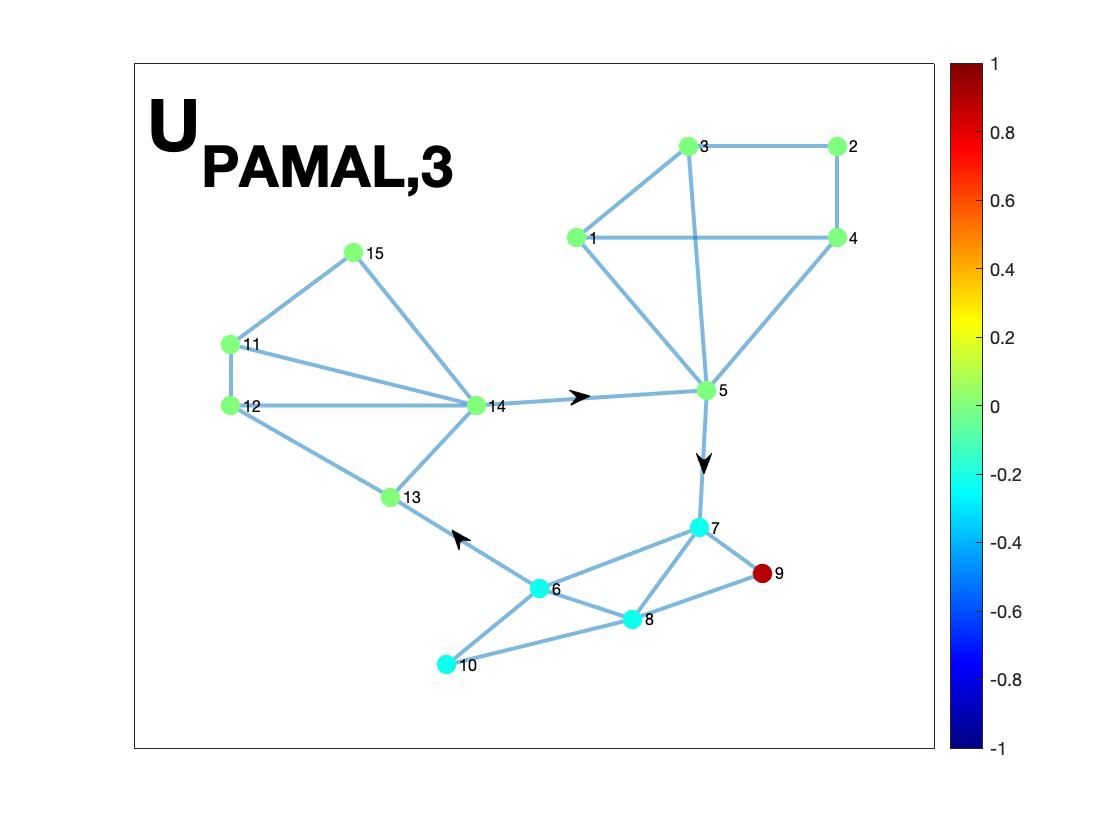}\\
\caption{Plotted on the top left is a directed unweighted graph with three clusters of 5 knots connected with a  directed cycle \cite[Fig. 1(c)]{Sardellitti17},
and on  the top right are its 15  frequencies via the SOC, PAMAL, Jordan and the proposed SVD approach.
On the next four rows from left to right are left frequency components, right frequency components of our proposed  GFT,
the SOC and PAMAL frequency components associated with $i$-th frequency, where $i=0,  1, 2, 3$ from top to bottom.
  }
\label{directedcyclefrequencies.fig}
  \end{center}
\end{figure}

For $0\le i\le N-1$, we obtain from the SVD \eqref{svd.def} that
${\bf u}_i$ and ${\bf v}_i$ in \eqref{UV.component} are the left and right eigenvectors of  ${\bf L}^T{\bf L}$ associated with the eigenvalue $\sigma_i^2$, i.e.,
\begin{equation}\label{leftrighteigenvaue.eq}
{\bf u}_i^T {\bf L}^T{\bf L}=\sigma_i^2 {\bf u}_i^T\ \ {\rm and}\  \  {\bf L}^T{\bf L} {\bf v}_i=\sigma_i^2 {\bf v}_i,\  0\le i\le N-1,
\end{equation}
\cite{magoarou2018, emirov2022}.
Then we call ${\bf u}_i$ and ${\bf v}_i, 0\le i\le N-1$, as the  {\em left  and  right  frequency components} associated with frequency $\sigma_i$, or $i$-th left (right) frequency components in short,  respectively.
 By \eqref{Laplacian.def0},
 the right frequency component associated with frequency zero can be selected as follows,
\begin{equation}\label{zerofrqeuncycomponent.right}
{\bf v}_0= N^{-1/2} {\bf 1}.
\end{equation}
The   left frequency component  ${\bf u}_0$ associated with frequency zero is not  always  a multiple of the constant signal ${\bf 1}$.  One may verify that
it can be so chosen that
\begin{equation}\label{zerofrqeuncycomponent.right}
{\bf u}_0= N^{-1/2} {\bf 1}
\end{equation}
if and only if
${\mathcal G}$ is an  Eulerian graph, in which
the in-degree and out-degree are the same at each vertex. 

In the undirected graph setting, the left and right  frequency components can be selected as the same, and they can be obtained via solving a family of constrained optimization problems inductively, 
\begin{eqnarray}\label{minmaxprinciple.eq00}
{\bf u}_i={\bf v}_i & \hskip-0.08in  = & \hskip-0.08in \arg \min_{{\bf x}\in W_i^\perp \ {\rm with} \ \|{\bf x}\|_2=1}  {\rm QV}({\bf x}) \nonumber\\
 & \hskip-0.08in  = & \hskip-0.08in \arg \min_{{\bf x}\in W_i^\perp \ {\rm with} \ \|{\bf x}\|_2=1}\|\L{\bf x}\|_2, 
\end{eqnarray}  with  the initial ${\bf v}_0=N^{-1/2}\bf 1$,
where   quadratic variation ${\rm QV}({\bf x})$ of a graph signal ${\bf x}$ is given in  \eqref{TV.def00}, 
and  for $1\le i\le N-1$, $W_i^\perp$ is the orthogonal complement of the space spanned by
${\bf v}_j, 0\le j\le i-1$.
Denote the average  and standard deviation of a vector ${\bf x}\in {\mathbb R}^N$ by
$$m({\bf x})=N^{-1} {\bf 1}^T{\bf x} \ \ {\rm  and} \ \ {\rm SD}({\bf x})= N^{-1/2} \|{\bf x}-m({\bf x}){\bf 1}\|_2,$$
the null space  of the transpose of Laplacian ${\bf L}$
by ${\rm ker}({\bf L}^T)$, and the dimension of a linear space $W$ by $\dim W$.
Based on the standard algorithm to find SVD and Courant-Fischer-Weyl min-max principle, we can apply the  following approach to construct
frequencies $\sigma_i$ and frequency components ${\bf v}_i$ and  ${\bf u}_i, 0\le i\le N-1$, of the proposed GFT:
 \begin{subequations}\label{frequencycomponents.eq}
\begin{equation}\label{frequencycomponents.eq0}
 \sigma_0=0 \ \ {\rm and }\ \
  {\bf v}_0 = N^{-1/2} {\bf 1}
\end{equation}
for $i=0$,
and
\begin{equation}\label{frequencycomponents.eq1}
\left\{\begin{array}{rcl}
\sigma_i &  \hskip-0.08in  = & \hskip-0.08in  \min_{{\bf x}\in W_i^\perp \ {\rm with} \ \|{\bf x}\|_2=1}\|\L{\bf x}\|_2\\
 & \hskip-0.08in = & \hskip-0.08in \min_{\dim W=i+1} \max_{{\bf x}\in W, \|{\bf x}\|_2=1} \|{\bf L}{\bf x}\|_2\\
& \hskip-0.08in = & \hskip-0.08in
\max_{\dim W=N-i} \min_{{\bf x}\in W, \|{\bf x}\|_2=1} \|{\bf L}{\bf x}\|_2\\
{\bf v}_i  & \hskip-0.08in  = & \hskip-0.08in  \arg \min_{{\bf x}\in W_i^\perp \ {\rm with} \ \|{\bf x}\|_2=1}\|\L{\bf x}\|_2
\end{array}\right.\end{equation}
inductively for $1\le i\le N-1$,  and let
${\bf u}_i, 0\le i\le i_0$, be an
orthonormal basis of the null space ${\rm ker}({\bf L}^T)$   with
 \begin{equation} \label{frequencycomponents.eq2}
  {\bf u}_0={\rm arg} \min_{{\bf x}\in {\rm ker}({\bf L}^T) \ {\rm with}\
 \|{\bf x}\|_2=1 \ {\rm and}\ m({\bf x})\ge 0} {\rm SD}({\bf x}),
 \end{equation}
  and define
\begin{equation} \label{frequencycomponents.eq3}
{\bf u}_i= \sigma_i^{-1} {\bf L} {\bf v}_i,\  i_0<i\le N-1,
\end{equation}
 \end{subequations}
 where  $i_0$ is the largest index such that $\sigma_{i_0}= 0$.
We remark that the left frequency component
 ${\bf u}_0$ associated with zero frequency in the above construction satisfies
 \eqref{zerofrqeuncycomponent.right}
if ${\mathcal G}$ is  an Eulerian graph,  and that $i_0=0$ if  the Laplacian ${\mathbf L}$ has rank $N-1$, or equivalently
if the graph  ${\mathcal G}$ is connected.
Shown in Figure \ref{directedcyclefrequencies.fig}
are  frequencies and frequency components of a  directed unweighted graph of size $15$ containing three clusters connected by a directed cycle \cite[Fig. 1(c)]{Sardellitti17}.
We observe that  frequency components with low frequencies may have certain clustering property and oscillation  pattern
related to the graph topology.

In addition to the quadratic variation ${\rm QV} ({\bf x})$
in \eqref{TV.def00}  and
$\|{\bf L} {\bf x}\|_2$ in \eqref {frequencycomponents.eq}, 
several directed variations have been proposed to measure the variation of a graph signal ${\bf x}=(x_i)_{i\in V}$
 along  the directed graph structure, including
\begin{equation}\label{dv.eq12}
{\rm GDV}({\bf x})=\sum_{i,j\in V}a_{ji}(x_i-x_j)_+
\end{equation}
 and
\begin{equation}\label{dv.eq11}
{\rm DV}({\bf x})=\sum_{i,j\in V}a_{ji}((x_i-x_j)_+)^2
\end{equation}
where  weight $a_{ij}$ is the $(i,j)$-th entry of the  adjacent matrix $\A$ and $t_+=\max(t, 0)$ for any real number $t\in {\mathbb R}$ \cite{Sardellitti17, Shafipour19}.
 We finish this section with some
 comparisons among  the  GFT in Definition
\ref{fourier.def} and the GFTs
 in
 \cite{Sardellitti17, Shafipour19}.

\begin{remark}\label {comparison.remark}
{\rm  In \cite{Sardellitti17}, the authors use the directed variation
${\rm GDV}({\bf x})$  in \eqref{dv.eq12}
as Lov\'{a}sz extension of the cut size function,
and define the GFT with
 frequency components ${\bf v}_i$ and frequencies
 $\lambda_i={\rm GDV}({\bf v}_i), 0\le i\le N-1$, being ordered so that $\lambda_0\le\lambda_1\le \ldots\le \lambda_{N-1}$,
where
${\bf V}=[{\bf v}_0, \ldots, {\bf v}_{N-1}]$ is the solution of the following constrained  minimization problem 
\begin{equation}\label{comparison.remark.eq1}
\min_{\bf V} \sum_{i=0}^{N-1}{\rm GDV}({\bf v}_i)
\end{equation}
subject to  ${\bf V}^T{\bf V}={\bf I}$  and  ${\bf v}_0=N^{-1/2} {\bf 1}$.
 To deal with the nonsmooth objective function and  non-convex orthogonality constraints in \eqref{comparison.remark.eq1}, the authors present
 two iterative algorithms,
 splitting orthogonality constraints (SOC for abbreviation) and proximal alternating minimization augmented Lagrange  (PAMAL for abbreviation),  to solve relaxed versions of the  constrained minimization problem \eqref{comparison.remark.eq1}, see
 \cite[Algorithms 1,  2, 3]{Sardellitti17}.
The above two implementations are more numerically stable than the method \eqref{Jordan.gft.def} based on the Jordan decomposition of Laplacian, however they
may fail to describe
different modes of variation over the directed graph.
Compared with the  GFT proposed in this paper where only the SVD of the   Laplacian  matrix of size $N\times N$ is required,
it needs to perform  SVD of a matrix of size $N\times N$ at each  iteration step of the iterative SOC and PAMAL  algorithms.

In \cite{Shafipour19}, the authors use the directed variation
${\rm DV}({\bf x})$  in \eqref{dv.eq11}  to measure the signal variation along the graph structure,
and define the GFT with
 frequency components ${\bf v}_i$ and frequencies
 $\lambda_i={\rm DV}({\bf v}_i), 0\le i\le N-1$, being ordered so that $\lambda_0\le\lambda_1\le \ldots\le \lambda_{N-1}$,
where
${\bf V}=[{\bf v}_0, \ldots, {\bf v}_{N-1}]$ is the solution of the following constrained  problem, 
\begin{equation}\label{comparison.remark.eq2a}
\min_{\bf V} \sum_{i=1}^{N-1}|{\rm DV}({\bf v}_{i})-{\rm DV}({\bf v}_{i-1})|^2
\end{equation}
subject to
${\bf V}^T{\bf V}={\bf I}$, ${\bf v}_0=N^{-1/2} {\bf 1}$ and
$
{\bf v}_{N-1}=\arg \max_{\|{\bf v}\|_2=1} {\rm DV}({\bf v})$.
Based on  the feasible method for optimization
 over the Stiefel manifold
in \cite{wen2013}, the authors develop an iterative algorithm
 to solve the  constrained problem  \eqref{comparison.remark.eq2a},
see  \cite[Algorithms 1 and 2]{Shafipour19}.
At each iteration, the proposed algorithm involves a matrix inversion and the computational  complexity is about $O(N^3)$.
Also as mentioned in \cite[Remark 1]{Shafipour19},
for the directed cycle graph (the circulant graph ${\mathcal C}_d(N, Q)$ generated by $Q=\{1\}$),
the proposed GFT in \cite{Shafipour19} fails to obtain the discrete Fourier transform in \eqref{dft.def}, cf. Theorem \ref{circulant.thm}
for our  SVD-based GFT in the directed circulant graph setting.
}
\end{remark}

\section{Graph Fourier transform on directed Eulerian  graphs}
\label{Euleriangraph.section}

Let ${\mathcal E}=(V, E)$ be an  Eulerian graph of order $N$ containing no loops or multiple edges,
 and
${\mathcal E}_t, 0\le t\le 1$, be  a   family of directed Eulerian graphs
that
 share the same  vertex set $V$ with the  graph ${\mathcal E}$ and
 have adjacent matrices ${\bf A}_t=(1-t){\bf A}+ t{\bf A}^T$ being  linear combinations of the adjacent matrices of the graph ${\mathcal E}$
and its transpose graph ${\mathcal E}^T$.
In this section, we consider  frequencies,  frequency components and  graph Fourier transforms
 on Eulerian graphs
${\mathcal E}_t, 0\le t\le 1$,
to connect the graph ${\mathcal E}$ and its transpose graph ${\mathcal E}^T$.
It is observed that frequencies and frequency components on the Eulerian graphs ${\mathcal E}_t, 0\le t\le 1$,
have certain symmetric properties, see \eqref{Euleriangraph.eq6} and Theorem \ref{Euleriangraph.cor}.
We say that frequencies $\sigma_i(t), 0\le i\le N-1$, of the Eulerian graphs ${\mathcal E}_t, 0\le t\le 1$,
are {\em simple}
if
\begin{equation} \label{Euleriangraph.thm.eq0}
0=\sigma_0(t)<\sigma_1(t)<\ldots< \sigma_{N-1}(t), \ 0\le t\le 1.
\end{equation}
In Theorem \ref{Euleriangraph.thm}, we show that
frequencies and frequency components are differentiable about $0\le t\le 1$,
if   frequencies  of the Eulerian graphs ${\mathcal E}_t, 0\le t\le 1$,
are simple.
%
To quantify and
measure the degree of asymmetry of the Eulerian graph ${\mathcal E}$, we define
\begin{equation} \label{Euleriangraph.eq7}
\sigma_{\rm asym}=\max_{\|{\bf x}\|_2=1} \|({\mathbf L}-{\bf L}^T){\bf x}\|_2,
\end{equation}
which is the same as the largest singular value  of ${\mathbf L}-{\bf L}^T$
 \cite{Li12}.
From the estimation in Theorem \ref{Euleriangraph.thm}, we  conclude that frequencies and frequency components have slow  variations to $0\le t\le 1$ when
$\sigma_{\rm asym}$  is small, see \eqref{perturbation.eq00} and \eqref{perturbation.eq01}.

Recall that an Eulerian graph  ${\cal E}$ has the same in-degree and out-degree  at each vertex, the Laplacians  ${\bf L}_t$ of the graphs ${\mathcal E}_t$
are given
by
\begin{equation}  \label{Euleriangraph.eq2}
{\mathbf L}_t= (1-t) {\bf L}+t {\bf L}^T,
\end{equation}
and satisfy
 \begin{equation}\label{Euleriangraph.lap.prop1}
{\mathbf L}_t{\bf 1}={\bf 0},  \  0\le t\le 1.
\end{equation}
By the continuity of  the Laplacian ${\mathbf L}_t, 0\le t\le 1$,
we can find an  SVD
 \begin{equation}  \label{Euleriangraph.eq3}
 \L_t =\U_t {\pmb\Sigma}_t\V_t^T
 \end{equation}
with initials $({\bf U}_0, {\bf V}_0, {\pmb \Sigma}_0)=({\bf U}, {\bf V}, {\pmb \Sigma})$
such that
orthogonal matrices ${\bf U}_t, {\bf V}_t$ and diagonal matrices
\begin{equation}
\label{Euleriangraph.eq4}
{\pmb\Sigma}_t={\rm diag}(\sigma_0(t), \ldots, \sigma_{N-1}(t))\end{equation}
 of singular values of Laplacians ${\mathbf L}_t$  in a nondecreasing order are continuous  about $0\le t\le 1$,
 where   ${\mathbf L}=\U{\pmb\Sigma}\V^T$ is the SVD \eqref{svd.def}
of the Laplacian ${\mathbf L}$.
 Using the above SVD of ${\mathbf L}_t$, we can define GFT
 ${\mathcal F}_t$ of a  signal ${\bf x}$ on the graph ${\mathcal E}_t$ (and also on ${\mathcal E}={\mathcal E}_0$ as they have the same vertex set) by
  \begin{equation}  \label{Euleriangraph.eq4+}
  {\mathcal F}_t{\bf x}= \frac{1}{2} \begin{pmatrix} {\bf U}_t^T{\bf x}+ {\bf V}_t^T {\bf x}\\
 {\bf U}_t^T{\bf x}- {\bf V}_t^T {\bf x}
 \end{pmatrix}, \ 0\le t\le 1.
 \end{equation}

 By \eqref{Euleriangraph.lap.prop1}, we have
 \begin{equation} \label{Euleriangraph.eq5}
 \sigma_0(t)=0, \ 0\le t\le 1.
 \end{equation}
 By the SVD \eqref{Euleriangraph.eq3}, $(\sigma_i(t))^2, 0\le i\le N-1$, are
 eigenvalues of matrices
${\mathbf L}_t^T {\mathbf L}_t$ and
${\mathbf L}_t{\mathbf L}_t^T$.  This together with the nonnegative nondecreasing order of singular values $\sigma_i(t), 0\le i\le N-1$, and the observation
that $ {\mathbf L}_{1-t}{\mathbf L}_{1-t}^T= {\mathbf L}_t^T {\mathbf L}_t, \ 0\le t\le 1$, proves that
\begin{equation} \label{Euleriangraph.eq6}
\sigma_i(1-t)=\sigma_i(t), \ 0\le t\le 1.
\end{equation}
Shown in  Figure \ref{euler_freq.fig}
are the graph frequencies  $\sigma_i(t), 0\le i\le N-1$, of Eulerian graphs  ${\mathcal E}_t, 0\le t\le 1$ of order $N=64$.
\begin{figure}[t]
\centering
\includegraphics[width=63mm, height=58mm]{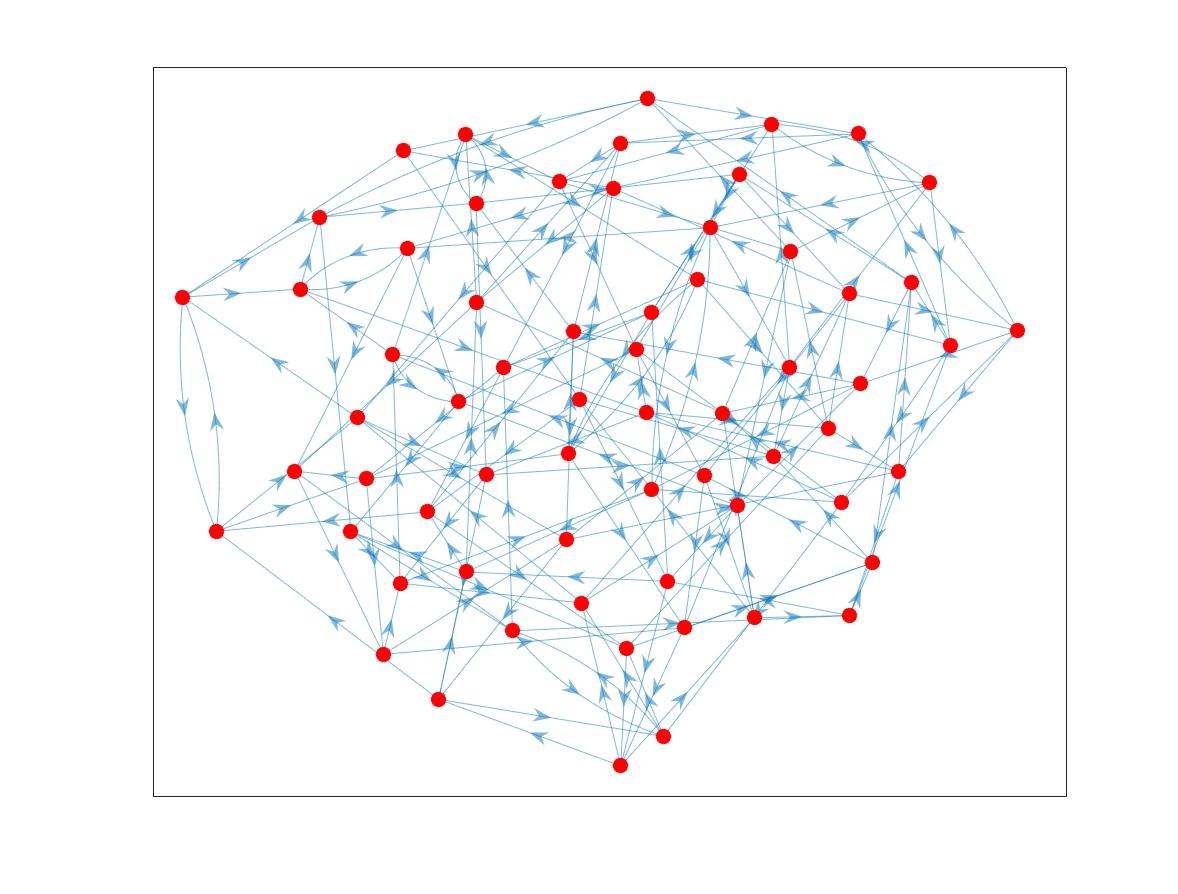}
\includegraphics[width=63mm, height=58mm]{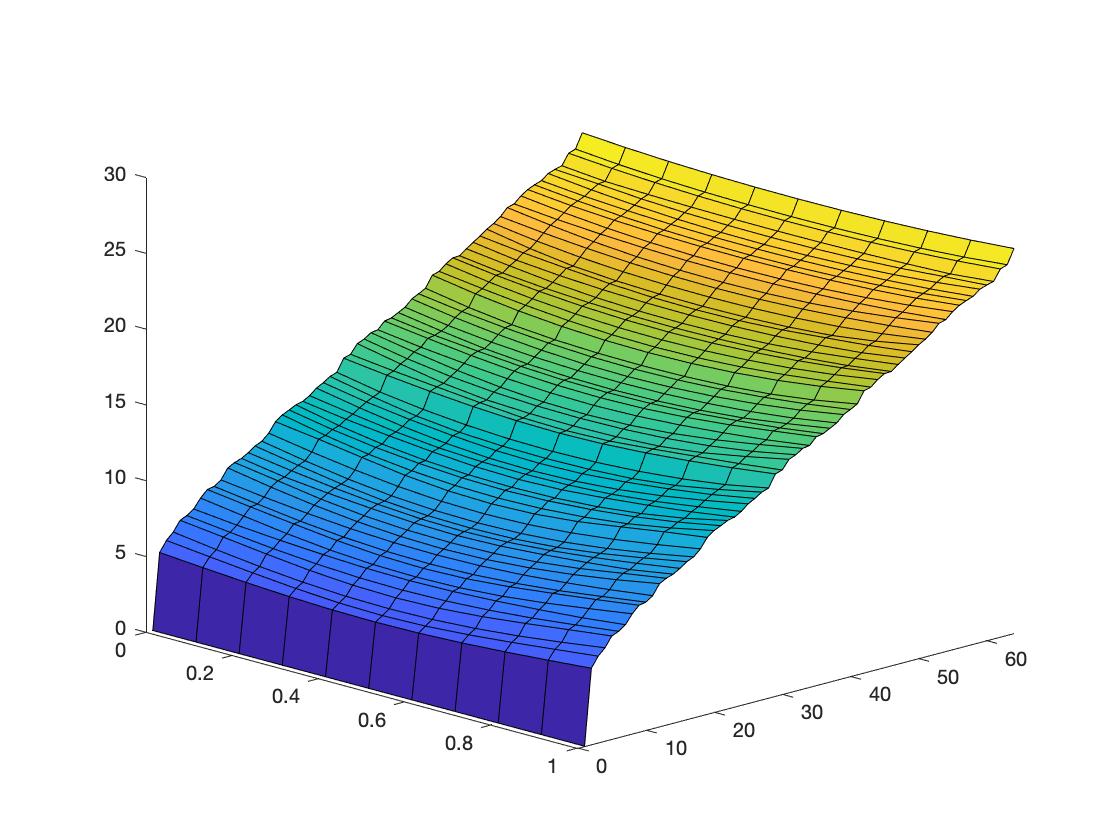}
\caption{Plotted on the left is an Eulerian graph of order $N=64$  with the associated Laplacian  $\L$  being a double stochastic matrix
with $\sigma_{\rm asym}=20.2248$. On the right is
 the frequencies $\sigma_i(t), 0\le i\le N-1$, of graph Laplacian matrices $\L_t=(1-t)\L+t\L^T$ with their maximal variation
 $\max_{0\le i\le N-1} \max_{0\le t\le 1}|\sigma_i(t)-\sigma_i(1/2)|=1.6750$.
 }\label{euler_freq.fig}
\end{figure}

Observe from \eqref{Euleriangraph.eq2} that
${\mathbf L}_t, 0\le t\le 1$, satisfy
\begin{equation} \label{Euleriangraph.eq8}
{\mathbf L}_t- {\mathbf L}_s= (t-s) ({\mathbf L}^T-{\mathbf L}), \ 0\le t, s\le 1.
\end{equation}
This together with the Courant-Fischer-Weyl
min-max principle,
\begin{eqnarray*} \sigma_i(t) & \hskip-0.08in = & \hskip-0.08in \min_{\dim W=i+1} \max_{{\bf x}\in W, \|{\bf x}\|_2=1} \|{\bf L}_t{\bf x}\|_2\\
& \hskip-0.08in = & \hskip-0.08in
\max_{\dim W=N-i} \min_{{\bf x}\in W, \|{\bf x}\|_2=1} \|{\bf L}_t{\bf x}\|_2
\end{eqnarray*}
implies that
$\sigma_i(t), 0\le i\le N-1$, are Lipschitz functions,
\begin{equation}\label{sigmalipschitz.eq}
|\sigma_i(t)-\sigma_i(s)|\le \sigma_{\rm asym} |t-s|,\ 0\le t, s\le 1.
\end{equation}
 In the
following theorem, we consider the  differentiability of frequencies
and left/right frequency components  with respect to $0\le t\le 1$, when  $\sigma_i(t), 0\le i\le N-1$, are  simple.
By \eqref{sigmalipschitz.eq}, we see that the simple requirement \eqref{Euleriangraph.thm.eq0}  is met if
all eigenvalues of Laplacian on ${\mathcal E}_{1/2}$ are simple,
and
the directed  Eulerian graph ${\mathcal E}$ is close to its undirected counterpart  ${\mathcal E}_{1/2}$ in the sense that
$$ 0<\sigma_{\rm asym}\le \alpha  \min_{1\le i\le N-1} \sigma_i(1/2)- \sigma_{i-1}(1/2)$$
 for some $0<\alpha<1$.

\begin{theorem}\label{Euleriangraph.thm} {\rm
Let 
${\mathcal E}_t, 0\le t\le 1$,
be the family of  directed  Eulerian graphs   to connect a directed Eulerian graph ${\mathcal E}$ and its transpose graph ${\mathcal E}^T$, and
 the associated  Laplacian $\L_t$  in \eqref{Euleriangraph.eq2}
has the SVD \eqref{Euleriangraph.eq3} with
 orthogonal matrices
${\bf U}_t=[{\bf u}_0(t), \ldots, {\bf u}_{N-1}(t)]$ and ${\bf V}_t=[{\bf v}_0(t), \ldots, {\bf v}_{N-1}(t)], \ 0\le t\le 1$
being continuous about $0\le t\le 1$ and satisfying
\begin{equation} \label{Euleriangraph.thm.eq5}
{\bf u}_0(t)={\bf v}_0(t)=N^{-1/2} {\bf 1}.\end{equation}
Then for any $1\le i\le N-1$,  the $i$-th frequency and frequency components  of the graph Fourier transform
${\mathcal F}_t, 0\le t\le 1$ is differentiable about $t$ if
it is a simple singular value, i.e., $\sigma_i(t)\ne \sigma_j(t)$ for all $j\ne i$. Moreover,
for all $1\le i\le N-1$,
\begin{equation}\label{Euleriangraph.thm.eq1}
\frac{d \sigma_i(t)}{dt}= ({\bf v}_i(t))^T ({\bf L}^T-{\bf L}){\bf u}_i(t),
\end{equation}
\begin{eqnarray} \label{Euleriangraph.thm.eq2}
\frac{d {\bf u}_i(t)}{dt}=\sum_{k=1}^{N-1}
\Big( - b_{i,k}(t) ({\bf v}_k(t))^T ({\bf L}^T-{\bf L}){\bf u}_i(t)  
+a_{i,k}(t)
({\bf u}_k(t))^T ({\bf L}^T-{\bf L}){\bf v}_i(t)\Big) {\bf u}_k\quad
\end{eqnarray}
and
\begin{eqnarray} \label{Euleriangraph.thm.eq3}
\frac{d {\bf v}_i(t)}{dt}=\sum_{k=1}^{N-1}
\Big( - a_{i,k}(t) ({\bf v}_k(t))^T ({\bf L}^T-{\bf L}){\bf u}_i(t) 
+b_{i,k}(t)
({\bf u}_k(t))^T ({\bf L}^T-{\bf L}){\bf v}_i(t)\Big) {\bf v}_k,\quad
\end{eqnarray}
where
$$a_{i, k}(t)= \left\{\begin{array}{ll} \sigma_i(t) ((\sigma_i(t))^2- (\sigma_{k}(t))^2)^{-1}  &  {\rm if} \ k\ne i\\
 (4 \sigma_i(t))^{-1} & {\rm if} \ k=i,
 \end{array}
 \right.$$
and
$$b_{i, k}(t)= \left\{\begin{array}{ll} \sigma_{k}(t) ((\sigma_i(t))^2- (\sigma_{k}(t))^2)^{-1}  &  {\rm if} \ k\ne i\\
 (-4 \sigma_i(t))^{-1} & {\rm if} \ k=i.
 \end{array}
 \right. $$
}
\end{theorem}

The detailed proof of Theorem \ref{Euleriangraph.thm} will be given in  Appendix \ref{Euleriangraph.thm.pfappendix}.
By Theorem \ref{Euleriangraph.thm}, we have
\begin{equation}\label{perturbation.eq00}
\left|\frac{d \sigma_i(t)}{dt}\right|\le \sigma_{\rm asym}.
\end{equation}
Set
$$C(t)=\max_{1\le i, k\le N-1} |a_{i, k}(t)|+ \max_{1\le i, k\le N-1} |b_{i, k}(t)|.$$
The orthogonality of the matrices $\bf U$ and $\bf V$ implies that 
\begin{eqnarray*}
\left\|\frac{d {\bf u}_i(t)}{dt}\right\|_2\hskip-0.05in&\le\hskip-0.05in&\hskip-0.05in\max_{1\le i, k\le N-1} |a_{i, k}(t)|\ \|({\bf L}^T-{\bf L}){\bf u}_i(t)\|_2
+\max_{1\le i, k\le N-1} |b_{i, k}(t)|\ \| ({\bf L}^T-{\bf L}){\bf v}_i(t)\|_2 
\le C(t)\sigma_{\rm asym}.
\end{eqnarray*}
Following a similar argument to $\left\|\frac{d {\bf v}_i(t)}{dt}\right\|_2$,  we have
\begin{equation}\label{perturbation.eq01}
\max \left(\left\|\frac{d {\bf u}_i(t)}{dt}\right\|_2, \left\|\frac{d {\bf v}_i(t)}{dt}\right\|_2\right)\le  C(t)\sigma_{\rm asym}.
\end{equation}
 This concludes that frequencies and frequency components have small variation about $0\le t\le 1$ when
  the degree   $\sigma_{\rm asym}$ of asymmetry of the Eulerian graph ${\mathcal E}$
 is small.

Under the simplicity assumption \eqref{Euleriangraph.thm.eq0} for all singular values $\sigma_i(t), 0\le i\le N-1$,
in addition to the symmetry \eqref{Euleriangraph.eq6} for the graph frequencies, we have the certain symmetric property for the  orthogonal matrices
 ${\bf U}_t$ and  ${\bf V}_t, 0\le t\le 1$, see Appendix \ref{Euleriangraph.cor.pfappendix} for the detailed proof.

\begin{theorem} \label{Euleriangraph.cor}
{\rm Let 
the family of  directed  Eulerian graph
${\mathcal E}_t, 0\le t\le 1$,
 the associated  Laplacian $\L_t$,
  the singular value decomposition
  ${\bf L}_t={\bf U}_t{\pmb \Sigma}_t{\bf V}_t$
be as in Theorem \ref{Euleriangraph.thm}.
If the singular values
$\sigma_i(t), 0\le i\le N-1$, satisfy
\eqref{Euleriangraph.thm.eq0}, then  for all $0\le t\le 1$,
\begin{equation} \label{Euleriangraph.cor.eq}
{\bf U}_t={\bf V}_{1-t},
\end{equation}
and
\begin{equation}\label{Euleriangraph.cor.eq+}
{\mathcal  F}_{1-t} {\bf x}= \begin{pmatrix} {\bf I} & {\bf 0}\\ {\bf 0} & -{\bf I}\end{pmatrix}
{\mathcal F}_t {\bf x}
\end{equation}
where ${\bf x}$ is a graph signal on the Eulerian graph ${\mathcal E}$.}
\end{theorem}

For the case that ${\mathcal E}$ is an  undirected graph (hence an Eulerian graph),
the orthogonal matrices ${\bf U}_t$ and ${\bf V}_t, 0\le t\le 1$, in the singular value decomposition
\eqref{Euleriangraph.eq3}  can be chosen to be independent on $t$.
The converse is true as well, because  ${\bf U}_t={\bf U}_{1/2}={\bf V}_{1/2}={\bf V}_t, 0\le t\le 1$,
by the independence of orthogonal matrices ${\bf U}_t$ and ${\bf V}_t$ on $0\le t\le 1$ and the observation that the Eulerian graph ${\mathcal G}_{1/2}$  is undirected,
we have that ${\bf L}^T={\bf L}$ and ${\mathcal E}$ is undirected.
In the following theorem, we show that
$({\bf L}^T)^2={\mathbf L}^2$ is a necessary condition for any pair of orthogonal matrices
${\bf U}_t$ and ${\bf V}_t, 0\le t\le 1$, are identical, see Appendix \ref{missundirected.thm.pfsection} for the proof.

\begin{theorem}\label{missundirected.thm}
{\rm Let ${\bf U}_t, {\bf V}_t$  be the orthogonal matrices  in the singular value decomposition
\eqref{Euleriangraph.eq3} of  the Laplacian ${\mathbf L}_t$  in \eqref{Euleriangraph.eq2}. If there exists $t_0\ne t_1\in [0, 1]$ such that
\begin{equation}\label{missundirected.thm.eq1}
{\bf U}_{t_0}={\bf U}_{t_1}\ {\rm and} \ {\bf V}_{t_0} ={\bf V}_{t_1},
\end{equation}
then  $({\bf L}^T)^2={\mathbf L}^2$.}
\end{theorem}

\section{Numerical simulations}
\label{numerocal.section}

Graph Fourier transform should be designed   to decompose  graph signals  into different frequency components,
 to represent them by different modes of variation efficiently,
 and to  have energy of smooth graph signals concentrated mainly at low frequencies.
In this section, we demonstrate the performance of the proposed SVD-based GFT to denoise
 the hourly temperature data set
collected at 218 locations
in the United States on August 1st, 2010
via
the bandlimiting  ${\bf P}_M$ at the first $M$-frequencies \cite{ncjs22,  zeng2017, cheng2020}.
Here for the SVD-based GFT proposed in this paper, the  bandlimiting ${\bf P}_M$ of a graph signal ${\bf x}$ at the first $M$-frequencies is given by
$${\bf P}_M {\bf x}=\! {\mathcal F}^{-1} \begin{pmatrix}
{\bf \chi}_M & \! {\bf O}\\
{\bf O}&\! {\bf \chi}_M
\end{pmatrix} {\mathcal F}{\bf x}\! =\frac{1}{2}\!
\sum_{i=0}^{M-1} \langle {\bf x}, {\bf u}_i\rangle {\bf u}_i+ \langle {\bf x}, {\bf v}_i\rangle {\bf v}_i,
$$
where $\chi_M$ is a diagonal  matrix with the first $M$ diagonal entries taking value one and all others taking value zero,
and for $0\le i\le M-1$, the $i$-th left/right frequency components ${\bf u}_i, {\bf v}_i$ are $i$-th columns of orthogonal matrices ${\bf U}$ and ${\bf V}$ in the SVD \eqref{svd.def} respectively.
For the GFT defined by splitting orthogonality constraints (SOC) and proximal alternating minimized augmented Lagrangian (PAMAL),
the  bandlimiting  ${\bf P}_M$ of a graph signal ${\bf x}$ at the first $M$-frequencies is given by
$${\bf P}_M{\bf x}=\sum_{i=0}^{M-1}\langle {\bf x},  {\bf v}_i\rangle {\bf v}_i $$
where ${\bf v}_i, 0\le i\le M-1$, is the $i$-th frequency component in  \cite{Sardellitti17}.

Let the underlying graph $\cal G$ of the US weather data set have $218$ vertices representing  locations of weather stations,
edges given by 5-nearest neighboring stations in physical distances,
and    weights on the edges are randomly chosen in $[0.8, 1.2]$, and denote the US temperature measured in Fahrenheit
 on August 1st, 2010  by ${\bf x}(t_k), 1\le k\le 24$, see  \cite[Fig. 6]{cheng2020} for two snapshots of the data set.
Shown in Figure  \ref{denoise_us_temp.fig} are the denoising performances in ISNR and SNR to apply
the bandlimiting projection ${\bf P}_M$ to  the noisy observations
\begin{equation}
\label{noiseweather.eq}
{\bf y}(t_k)={\bf x}(t_k)+{\pmb \eta}(t_k), 1\le k\le 24,\end{equation}
  corrupted with additive  random noises ${\pmb \eta}(t_k)$ with entries being i.i.d. and having mean zero and variance
$c\in [4, 16]$.
Here
 the input signal-to-noise ratio
(ISNR) and
the output signal-to-noise ratio (SNR) are defined by
	$$
{\rm ISNR}= -20 \log_{10} \frac{\|{\pmb \eta}\|_2}{\|{\bf x}\|_2}\ {\rm and} \
{\rm SNR}=-20 \log_{10} \frac{\|\widehat {\bf x}-{\bf x}\|_2}{\|{\bf x}\|_2}, $$
 where  the original signal ${\bf x}$, the noisy measurement ${\bf y}$
 and the denoised signal $\widehat {\bf x}$ 
 are given by
 the hourly weather data ${\bf x}(t_k)$, the noisy weather data ${\bf y}(t_k)$ in \eqref{noiseweather.eq},
and the denoised signal ${\bf P}_M{\bf y}(t_k)$ by  bandlimiting ${\bf P}_M$
 the noisy weather data ${\bf y}(t_k)$  to the first $M$-frequencies
  respectively.
 We observe from Figure  \ref{denoise_us_temp.fig}
  that the SVD-based GFT proposed in this paper outperforms
 the SOC and PAMAL-based GFTs in \cite{Sardellitti17}
 on denoising the US hourly weather data set
 on August 1st, 2010 by bandlimiting  ${\bf P}_M$ at the first $M$-frequencies.
 It is also noticed that
 the SOC and PAMAL-based GFTs in \cite{Sardellitti17} have very similar performance on denoising the weather data set. The possible reason is that
 they  are based on  different relaxations of the same  constrained
minimization problem \eqref{comparison.remark.eq1}.

\begin{figure}
\centering
\includegraphics[width=52mm, height=48mm]{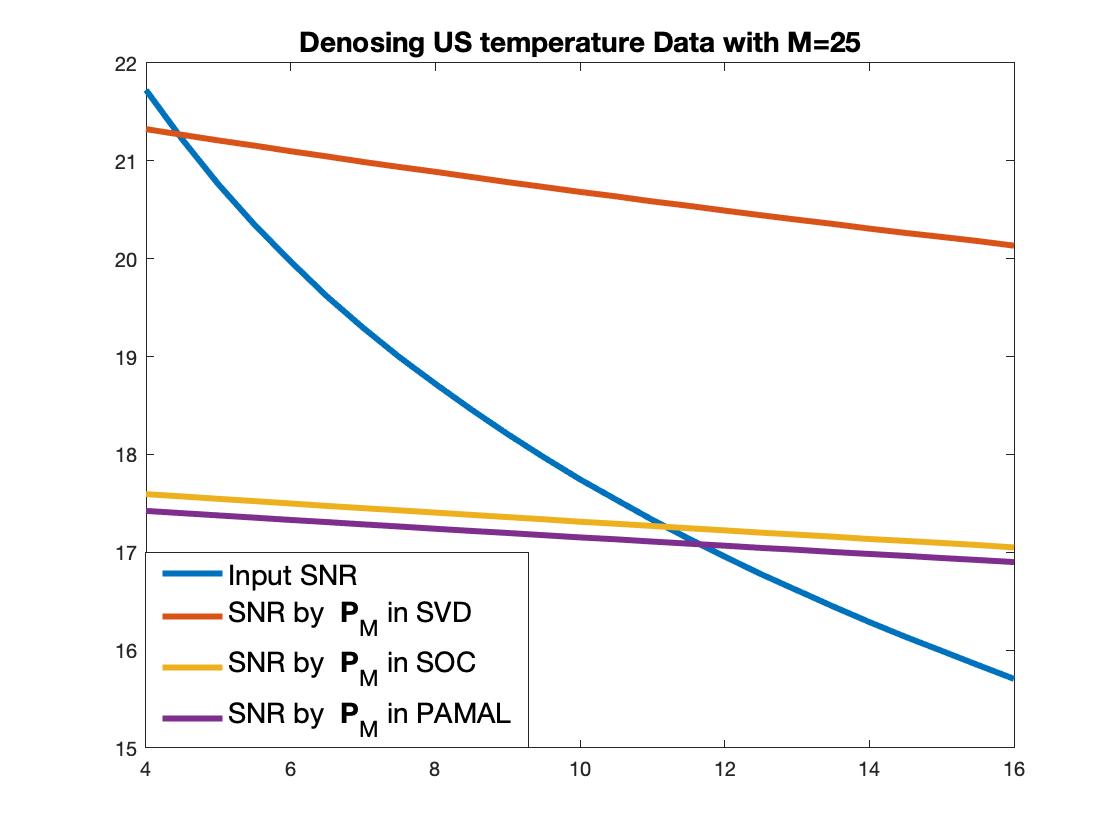}
\includegraphics[width=52mm, height=48mm]{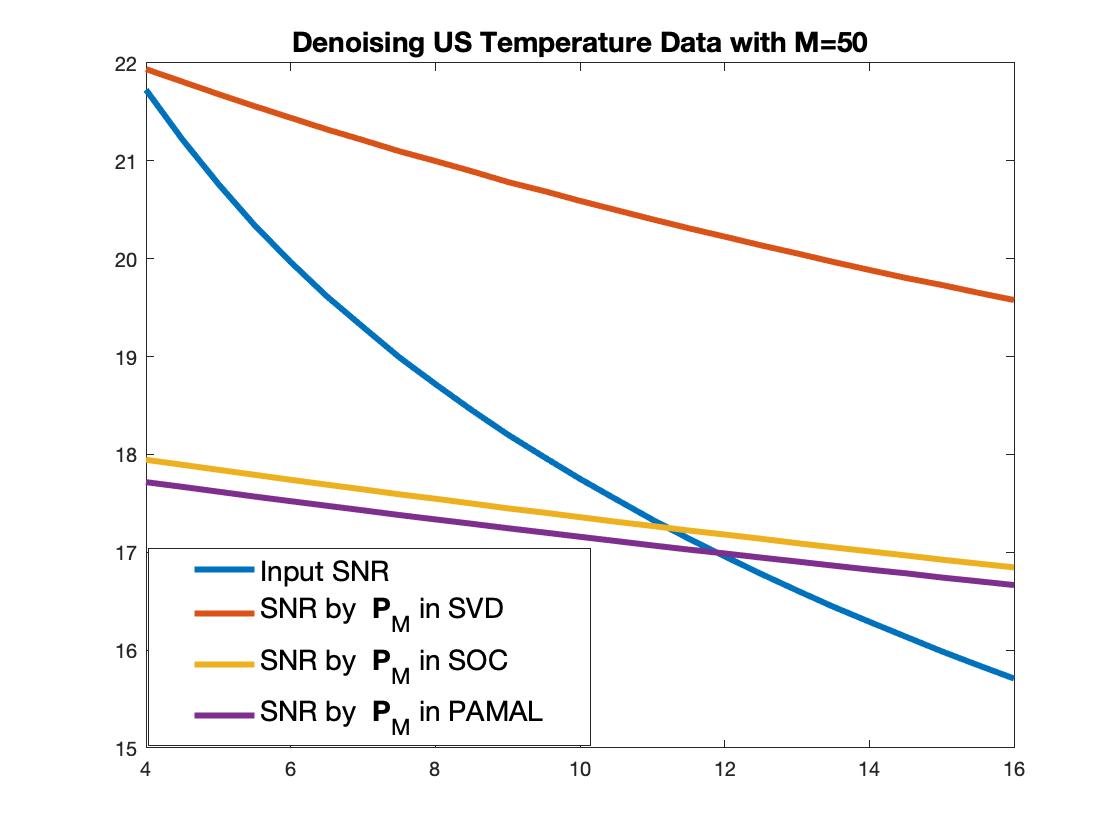}
\caption{
Plotted
are the averages of  ISNR and SNR of denoising the hourly  temperature data ${\bf y}(t_k), 1\le k\le 24$,  on August 1st, 2012 over $1000\times 24$ trials,  via
bandlimiting  ${\bf P}_M$ at the first $M$-frequencies
of the SVD, SOC and PAMAL-based GFTs,
where $M=25$  (left) and  $M=50$ (right).
  }\label{denoise_us_temp.fig}
\end{figure}

    \appendix

In the appendix, we collect the proofs of Theorems \ref{circulant.thm}, \ref{Euleriangraph.cor} and \ref{missundirected.thm}.

\subsection{Graph Fourier transform on circulant graphs}\label{TheoremCirculant.proof}

In this appendix, we consider GFT on circulant graphs and provide a  proof of  Theorem \ref{circulant.thm}.

Write $Q=\{q_1, \ldots, q_L\}$ with $1\le q_1<q_2<\ldots< q_L\le N-1$,
and ${\bf W}=[{\bf w}_0, \ldots {\bf w}_{N-1}]$.
Observe that the Laplacian matrix
${\mathbf L}_{{\mathcal C}_d}=(c_{ij})_{0\le i, j\le N-1}$
on the circulant graph ${\mathcal C}_d:={\mathcal C}_d(N, Q)$
is a circulant matrix with $ij$-th entries
$c_{ij}, 0\le i, j\le N-1$,
 given by
$$c_{ij}=\left\{ \begin{array}{ll} L  & {\rm if} \ j=i\\
-1 & {\rm if} \ j-i \in Q\   {\rm mod}\ N\\
0 & {\rm otherwise}.
\end{array} \right.
$$
Then  one may verify that
\begin{equation} \label{circulant.thm.pfeq1} %
{\mathbf L}_{{\mathcal C}_d} {\bf w}_i= P(\omega_N^{i}) {\bf w}_i,\ 0\le i\le N-1,
\end{equation}
where $P$ is the polynomial symbol of the circulant matrix ${\mathbf L}_{{\mathcal C}_d}$  defined by  \eqref{circulant.symbol}.

Let
\begin{equation} \label{circulant.thm.pfeq2}
{\bf M}=
{\rm diag} (|P(1)|, |P(\omega_N)|,  \ldots,  |P(\omega_N^{N-1})|) \end{equation}
be the diagonal matrix
 with magnitudes $|P(\omega_N^{i})|, 0\le i\le N-1$,
of the symbol $P$ on all $N$-th unit roots.
Then we can reformulate \eqref{circulant.thm.pfeq1}
in the following matrix form,
\begin{equation} \label{circulant.thm.pfeq3}
{\mathbf L}_{{\mathcal C}_d} {\bf W}= {\bf W} {\pmb\Theta}  {\bf M},
\end{equation}
where ${\pmb\Theta}$ is the diagonal matrix in \eqref{circulant.thm.eq2}.

 Let ${\bf Q}$ be a permutation matrix to rearrange
$|P(\omega_N^{i})|, 0\le i\le N-1$, in nondecreasing order, with $0$ as the first index,  and indices $i$ and $N-i, 1\le i<N/2$, next each other.
This together with $|P(\omega_N^{i})|= |P(\omega_N^{N-i})|, 1\le i\le N-1$, implies that
the diagonal matrix ${\pmb \Sigma}$ in \eqref{circulant.thm.eq5}
satisfies
\begin{equation}\label{circulant.thm.pfeq4}
{\pmb \Sigma}=  {\bf Q} {\bf M} {\bf Q}.
\end{equation}

Let   ${\bf e}_k, 0\le k\le N-1$, be the unit vectors with zero entries except the $k$-th entry taking value $1$, and
define  the permutation matrices  ${\bf P}_0$ and ${\bf P}_1$ by
\begin{equation}\label{circulant.thm.pfeq5}
{\bf P}_0=\left\{\begin{array}{l}
\! \big[{\bf e}_0,{\bf e}_1,{\bf e}_{N-1}, \ldots,{\bf e}_{(N-1)/2},{\bf e}_{(N+1)/2}\big]\qquad  \hfill {\rm if} \ N\ {\rm is\  odd}\\
\! \big[{\bf e}_0,{\bf e}_1,{\bf e}_{N-1}, \ldots,{\bf e}_{N/2-1},{\bf e}_{N/2+1}, {\bf e}_{N/2} \big]\qquad
\hfill {\rm if} \ N\ {\rm is\  even}
\end{array}
\right.
\end{equation}
and
\begin{equation}
\label{circulant.thm.pfeq6}
{\bf P}_1={\bf P}_0{\bf Q}.
\end{equation}
Therefore the conclusion in Theorem
 \ref{circulant.thm} about the GFT on the circulant graph ${\mathcal C}_d(N, Q)$ reduces to the
 singular value decomposition of ${\mathbf L}_{{\mathcal C}_d}$ in the following proposition.

\begin{proposition}\label{singular.prop}
{\rm
Let  ${\mathbf L}_{{\mathcal C}_d}$
be  the Laplacian matrix
on the circulant graph ${\mathcal C}_d:={\mathcal C}_d(N, Q)$,
and
$ {\bf W}, {\bf R},  {\pmb\Theta}, {\pmb \Sigma},  {\bf P}_0$ and ${\bf P}_1$
be as in \eqref{wdft.def}, \eqref{circulant.thm.eq1}, \eqref{circulant.thm.eq2}, \eqref{circulant.thm.eq5},
\eqref{circulant.thm.pfeq5} and \eqref{circulant.thm.pfeq6} respectively.
Then the matrices ${\bf U}$ and ${\bf V}$ in \eqref{circulant.thm.eq4} are orthogonal matrices with real entries, and
the singular value decomposition \eqref{circulant.thm.eq3} holds for  the Laplacian matrix
${\mathbf L}_{{\mathcal C}_d}$.
}
\end{proposition}

\begin{proof} The conclusions are trivial for $N=1$ and $N=2$. So we assume that $N\ge 3$ now.
 First we divide two cases, $N\ge 3$ is  odd and even,
to prove that matrices ${\bf U}$ and ${\bf V}$ in \eqref{circulant.thm.eq4} are orthogonal matrices with real entries.
Define
\begin{equation}\label{circulant.thm.pfeq10}
\widetilde {\bf U}={\bf W}{\pmb\Theta}{\bf P}_0{\bf R}=[\widetilde {\bf u}_{0},
\widetilde {\bf u}_{1},\ldots, \widetilde {\bf u}_{N-1}]
\end{equation}
and
\begin{equation}\label{circulant.thm.pfeq11}
\widetilde {\bf V}={\bf W} {\bf P}_0{\bf R}=[\widetilde {\bf v}_{0},
\widetilde {\bf v}_{1},\ldots, \widetilde  {\bf v}_{N-1}].
\end{equation}
As ${\bf P}_1$ is a permutation matrix, $\U=\widetilde {\bf U}{\bf P}_1$ and $\V=\widetilde {\bf V}{\bf P}_1$, it suffices to
show that
$\widetilde {\bf U}$ and $\widetilde {\bf V}$ are  orthogonal matrices with real entries.

\smallskip

{\em Case 1: \ $N=2K+1$  for some integer $K\ge 1$.}

By \eqref{circulant.thm.eq1},
\eqref{circulant.thm.pfeq10} and \eqref{circulant.thm.pfeq11},
we have
\begin{equation} \label{circulant.thm.pfeq12}
\widetilde {\bf u}_{0}
=\widetilde {\bf v}_{0}={\bf w}_{0}= N^{-1/2} {\bf 1}\in {\mathbb R}^N,
\end{equation}
and for $1\le k\le K$
\begin{equation}   \label{circulant.thm.pfeq13}
\left\{\begin{array}{l}
\widetilde {\bf v}_{2k-1}= \frac{{\bf w}_{k}+ {\bf w}_{N-k}}{\sqrt{2}}=\frac{{\bf w}_{k}+ \overline{\bf w}_{k}}{\sqrt{2}}\in {\mathbb R}^N\\
\widetilde  {\bf v}_{ 2k}= \frac{{\bf w}_{k}- {\bf w}_{N-k}}{\sqrt{-2}}=\frac{{\bf w}_{k}-\overline {\bf w}_{k}}{\sqrt{-2}}\in {\mathbb R}^N,
\end{array}\right.
\end{equation}
and
\begin{equation}   \label{circulant.thm.pfeq14}
\left\{\begin{array}{l}
\widetilde {\bf u}_{2k-1}=\frac{\exp(\sqrt{-1}\theta_k) {\bf w}_{k}+ \exp(-\sqrt{-1}\theta_k)\overline{{\bf w}_{k}}}{\sqrt{2}}\in {\mathbb R}^N\\
\widetilde {\bf u}_{2k}= \frac{\exp(\sqrt{-1}\theta_k){\bf w}_{k}-\exp(-\sqrt{-1}\theta_k) \overline {{\bf w}_{k}}}{\sqrt{-2}}\in {\mathbb R}^N.
\end{array}\right.
\end{equation}
Therefore $\widetilde {\bf U}$  and $\widetilde {\bf V}$ are  square matrices with real entries.
This together with the unitary property for the discrete Fourier transform matrix ${\bf W}$, the phase matrix ${\pmb\Theta}$ and the rotation matrix ${\bf R}$,
and the orthogonality of the permutation  matrix  ${\bf P}_1$ implies that
\begin{equation*} 
{\widetilde {\bf U}}^T{\widetilde {\bf U}}=
\widetilde {\bf U}^H\widetilde {\bf U}=  {\bf R}^H{\bf P}_0^T{\pmb\Theta}^H {\bf W}^H {\bf W}{\pmb\Theta} {\bf P}_0{\bf R}={\bf I}\end{equation*}
and
\begin{equation*} 
\widetilde {\bf V}^T\widetilde {\bf V}=
\widetilde {\bf V}^H\widetilde {\bf V}= {\bf R}^H{\bf P}_0^T{\bf W}^H {\bf W}{\bf P}_0{\bf R}={\bf I}.\end{equation*}
This proves that
$\tilde {\bf U}$ and $\widetilde {\bf V}$ (and hence ${\bf U}$ and ${\bf V}$ in \eqref{circulant.thm.eq4}) are  orthogonal matrices with real entries
for the case that $N$ is odd.

\smallskip
{\em  Case 2: \ $N=2K+2$  for some integer $K\ge 1$}. 

Using the similar argument used in Case 1, we can show that
\eqref{circulant.thm.pfeq12}, \eqref{circulant.thm.pfeq13}
and \eqref{circulant.thm.pfeq14} hold. In addition, we have
\begin{equation} \label{circulant.thm.pfeq15}
\widetilde {u}_{2K+1} =\widetilde {\bf v}_{2K+1}=N^{-1/2} (1, -1, \ldots,  1, -1)^T\in {\mathbb R}^N
\end{equation}
Therefore $\widetilde {\bf U}$  and $\widetilde {\bf V}$ are  square matrices with real entries.
The orthogonal property for the matrices  $\widetilde {\bf U}$  and $\widetilde {\bf V}$ can be established in a similar way used in Case 1.
This proves that
$\tilde {\bf U}$ and $\widetilde {\bf V}$ (and hence ${\bf U}$ and ${\bf V}$ in \eqref{circulant.thm.eq4}) are  orthogonal matrices with real entries
for the case that $N$ is even.

Next we establish the singular value decomposition \eqref{circulant.thm.eq3} for  the Laplacian matrix
${\mathbf L}_{{\mathcal C}_d}$.
By $|P(\omega_N^{i})|^2= |P(\omega_N^{N-i})|^2, 1\le i\le N-1$, one may verify that
\begin{equation}  \label{circulant.thm.pfeq16}
{\bf M}  {\bf P}_0 {\bf R}{\bf P}_0=  {\bf P}_0 {\bf R} {\bf P}_0 {\bf M}.\end{equation}
By \eqref{circulant.thm.eq4}, 
\eqref{circulant.thm.pfeq3}, \eqref{circulant.thm.pfeq4}, 
\eqref{circulant.thm.pfeq6},
  \eqref{circulant.thm.pfeq16}, and the permutation property ${\bf Q}^2={\bf I}$, we obtain
\begin{eqnarray*}
\L_{{\mathcal C}_d}{\bf V} & \hskip-0.08in  = & \hskip-0.08in  \L_{{\mathcal C}_d}{\bf W}{\bf P}_0 {\bf R}{\bf P}_1=
{\bf W} {\pmb\Theta} {\bf M}  {\bf P}_0 {\bf R} {\bf P}_0 {\bf Q}
={\bf W} {\pmb\Theta}  {\bf P}_0  {\bf R}  {\bf P}_0 {\bf M}   {\bf Q}
=
 {\bf U} {\pmb \Sigma}.
\end{eqnarray*}
This together with the real orthogonal property for the matrices ${\bf U}$ and ${\bf V}$ proves the
singular value decomposition in
 \eqref{circulant.thm.eq3} for  the Laplacian matrix
${\mathbf L}_{{\mathcal C}_d}$, and hence completes the proof.
\end{proof}

We finish this appendix with the proof of Theorem \ref{circulant.thm}.

\begin{proof} [Proof of Theorem \ref{circulant.thm}] Let ${\bf P}_0$ and ${\bf P}_1$ be as in
\eqref{circulant.thm.pfeq5}  and \eqref{circulant.thm.pfeq6} respectively.
 By  Proposition \ref{singular.prop}, we have
\begin{eqnarray*}
{\mathcal F}{\bf x} 
& \hskip-0.08in = & \hskip-0.08in
\frac{1}{2}
\begin{pmatrix}
 {\bf P}_1 {\bf R}^H {\bf P}_0 ({\pmb\Theta}^H+{\bf I}) {\bf W}^H {\bf x}\\
 {\bf P}_1 {\bf R}^H {\bf P}_0 ({\pmb\Theta}^H-{\bf I}) {\bf W}^H {\bf x}
\end{pmatrix}\nonumber\\
& \hskip-0.08in = & \hskip-0.08in
\frac{1}{2}
\begin{pmatrix}
 {\bf P}_1 {\bf R}^H {\bf P}_0 & {\bf 0}\\
 {\bf 0} & {\bf P}_1 {\bf R}^H {\bf P}_0 \end{pmatrix}  \begin{pmatrix} {\pmb\Theta}^H  & {\bf I}\\
{\pmb\Theta}^H  & -{\bf I}\end{pmatrix}
\begin{pmatrix} {\bf W}^H {\bf x}\\
 {\bf W}^H {\bf x}
\end{pmatrix}.
\end{eqnarray*}
This completes the proof.
\end{proof}

\subsection{Proof of Theorem \ref{Euleriangraph.thm}}\label{Euleriangraph.thm.pfappendix}
Let
${\mathcal S} ({\mathbf L}_t)$ be
 the  self-adjoint dilation  of the Laplacian ${\mathbf L}_t, 0\le t\le 1$.  By
 \eqref{Euleriangraph.eq3}, we have
\begin{equation}\label{Euleriangraph.thm.pfeq1}
{\cal S}(\L_t)
{\bf F}_t
={\bf F}_t\pmb\Lambda_t,
\end{equation}
where
\begin{equation} \label{Euleriangraph.thm.pfeq2}
{\bf F}_t=\frac{1}{\sqrt2}\begin{pmatrix}
\U_t&\U_t\\
\V_t&-\V_t
\end{pmatrix}=[{\bf z}_0(t),{\bf z}_1(t),\ldots,{\bf z}_{2N-1}(t)]
\end{equation}
 and
 \begin{equation} \label{Euleriangraph.thm.pfeq3}
 \pmb\Lambda_t:=\begin{pmatrix}
\pmb\Sigma_t&{\bf O}\\
{\bf O}&-\pmb\Sigma_t
\end{pmatrix}={\rm diag}(\lambda_0(t),\lambda_2(t),\ldots, \lambda_{2N-1}(t)).\end{equation}
By \eqref{sigmalipschitz.eq} and
the assumption on  $i$-th frequency $\sigma_i(t), 1\le i\le N-1$,  we can find $\delta>0$ such that for all  $0\le s\le 1$ with $|s-t|<\delta$,
\begin{equation} \label{Euleriangraph.thm.pfeq3+} \lambda_i(s)=\sigma_i(s)\end{equation}
is a simple eigenvalue of
self-adjoint dilation  ${\mathcal S}({\mathbf L}_s)$  of the Laplacian ${\mathbf L}_s, 0\le s\le 1$,
and
$ {\bf z}_i(s)
$
is an associated eigenvector with norm one.
This together with  \eqref{Euleriangraph.thm.pfeq1} and \eqref{Euleriangraph.thm.pfeq2}
 implies that
\begin{eqnarray}  \label{Euleriangraph.thm.pfeq4} & &
\begin{pmatrix}  \sigma_i(t){\bf I}-{\mathcal S}({\mathbf L}_t)   & {\bf z}_i(t)\\
{\bf z}_i(t)^T & 0\end{pmatrix}
 =   \begin{pmatrix} {\bf F}_t & {\bf 0}\\ {\bf 0} & 1\end{pmatrix}
\begin{pmatrix} \lambda_i(t) {\bf I}- {\pmb \Lambda}_t & {\bf e}_{i}\\
{\bf e}^T_{i}  & 0\end{pmatrix}
 \begin{pmatrix} {\bf F}_t & {\bf 0}\\ {\bf 0} & 1
\end{pmatrix}^T,\qquad
\end{eqnarray}
is nonsingular, where
 ${\bf e}_i, 0\le i\le 2N-1$, are  unit vectors of size $2N$ with all entries taking value zero except value one at $i$-th entry. 

Define a map $H: \R^{2N}\times\R\times[0,1]\rightarrow \R^{2N+1}$ by
\begin{equation}\label{H.def}
H({\bf z}, \lambda, t)=\begin{pmatrix}
\lambda  {\bf z}-{\mathcal S}(\L_t){\bf z}\\
\frac{1}{2}({\bf z}^T{\bf z}-1)
\end{pmatrix}.
\end{equation}
Then 
\begin{equation} \label{Euleriangraph.thm.pfeq6}
\nabla_{{\bf z}, \lambda} H({\bf z}, \lambda, t)=\begin{pmatrix}
\lambda {\bf I}-\mathcal S(\L_t)& {\bf z}\\
{\bf z}^T&0
\end{pmatrix}
\end{equation}
and
 \begin{equation} \label{Euleriangraph.thm.pfeq7}
\nabla_t H({\bf z}, \lambda, t)=\begin{pmatrix}
 \mathcal S({\mathbf L}^T-{\bf L}){\bf z}\\
0
\end{pmatrix}.
\end{equation}
By \eqref{Euleriangraph.thm.pfeq1}, \eqref{Euleriangraph.thm.pfeq3+}, \eqref{Euleriangraph.thm.pfeq4}, \eqref{Euleriangraph.thm.pfeq6}, \eqref{Euleriangraph.thm.pfeq7}
and the implicit function theorem, there exists $0<\tilde{\delta} <\delta$ such that for all $s$ with $|s-t|<\tilde \delta$,
$({\bf z}_i(s)^T, \sigma_i(s))$ is the unique solution of
$$H({\bf z}, \lambda, s)={\bf 0}$$
in the neighborhood of $({\bf z}_i(t)^T, \sigma_i(t))$. Applying the implicit function theorem again
and using \eqref{Euleriangraph.thm.pfeq4},  \eqref{Euleriangraph.thm.pfeq6},  \eqref{Euleriangraph.thm.pfeq7},
we obtain
\begin{eqnarray} \label{Euleriangraph.thm.pfeq8}
\begin{pmatrix}
\frac{d{\bf z}_i(t)}{dt}\\
\frac{d\sigma_i(t)}{dt}
\end{pmatrix} \hskip-0.05in & = \hskip-0.05in & \hskip-0.05in
- (\nabla_{{\bf z}, \lambda} H({\bf z}, \lambda, t))^{-1}  \nabla_t H({\bf z}, \lambda, t)\nonumber\\
\hskip-0.05in &=\hskip-0.05in&\hskip-0.05in
 - \begin{pmatrix} {\bf F}_t & {\bf 0}\\ {\bf 0} & 1\end{pmatrix}
\begin{pmatrix} (\sigma_i(t) {\bf I}- {\pmb \Lambda}_t)^\dag  & {\bf e}_i\\
{\bf e}^T_i  & 0\end{pmatrix}
 \begin{pmatrix} {\bf F}_t & {\bf 0}\\ {\bf 0} & 1
\end{pmatrix}^T
\begin{pmatrix}
 \mathcal S({\mathbf L}^T-{\bf L}){\bf z}_i(t)\\
0
\end{pmatrix}\nonumber\\
\hskip-0.05in&=\hskip-0.05in&\hskip-0.05in
-
 \begin{pmatrix}
  {\bf F}_t (\sigma_i(t) {\bf I}- {\pmb \Lambda}_t)^\dag {\bf F}_t^T {\mathcal S}(\L^T-\L) {\bf z}_i(t)\\
{\bf z}_i(t)^T {\mathcal S}(\L^T-\L) {\bf z}_i(t) \end{pmatrix},\qquad
\end{eqnarray}
where the diagonal matrix
$(\sigma_i(t) {\bf I}- {\pmb \Lambda}_t)^\dag$ is the pseudo-inverse of
$\sigma_i(t) {\bf I}- {\pmb \Lambda}_t$ with $j$-th diagonal entries being
 $(\sigma_i(t)-\sigma_j(t))^{-1}$ for $i\ne j\le N-1$,  $0$ for $j=i$ and
 $(\sigma_i(t)+\sigma_{j-N}(t))^{-1}$ for $N\le j\le 2N-1$.
Substituting  $ {\bf z}_i(t)=\frac{1}{\sqrt{2}} \begin{pmatrix} {\bf u}_i(t)\\ {\bf v}_i(t)\end{pmatrix}$
into \eqref{Euleriangraph.thm.pfeq8}, we prove
\eqref{Euleriangraph.thm.eq1}.

Let
\begin{eqnarray}\label{Euleriangraph.thm.pfeq9}
{\bf A}_{i, t} & \hskip-0.08in = & \hskip-0.08in  \frac{1}{2}\Big((\sigma_i(t) {\bf I}- {\pmb \Sigma}_t)^\dag+(\sigma_i(t) {\bf I}+ {\pmb \Sigma}_t)^{-1}\Big)
={\rm diag} (a_{i, 0}(t), \ldots, a_{i, N-1}(t))
\end{eqnarray}
and
\begin{eqnarray} \label{Euleriangraph.thm.pfeq10}
{\bf B}_{i, t} & \hskip-0.08in = & \hskip-0.08in \frac{1}{2}\Big((\sigma_i(t) {\bf I}- {\pmb \Sigma}_t)^\dag-(\sigma_i(t) {\bf I}+ {\pmb \Sigma}_t)^{-1}\Big)
={\rm diag} (b_{i, 0}(t), \ldots, b_{i, N-1}(t)).
\end{eqnarray}
  By \eqref{Euleriangraph.thm.pfeq2} and \eqref{Euleriangraph.thm.pfeq8}, we obtain
\begin{eqnarray} \label{Euleriangraph.thm.pfeq11}
\frac{d{\bf u}_i(t)}{dt} & \hskip-0.08in = & \hskip-0.08in
{\bf U}_t {\bf A}_{i, t} {\bf U}_t^T ({\bf L}^T-{\bf L}) {\bf v}_i(t)
 - {\bf U}_t {\bf B}_{i, t} {\bf V}_t^T ({\bf L}^T-{\bf L}) {\bf u}_i (t)
\end{eqnarray}
and
\begin{eqnarray} \label{Euleriangraph.thm.pfeq12}
\frac{d{\bf v}_i(t)}{dt} & \hskip-0.08in = & \hskip-0.08in
{\bf V}_t {\bf B}_{i, t} {\bf U}_t^T ({\bf L}^T-{\bf L}) {\bf v}_i(t)
- {\bf V}_t {\bf A}_{i, t} {\bf V}_t^T ({\bf L}^T-{\bf L}) {\bf u}_i (t).
\end{eqnarray}
By \eqref{Euleriangraph.thm.eq5}, we have
$${\bf u}_0^T ({\bf L}^T-{\bf L})= {\bf v}_0^T ({\bf L}^T-{\bf L})=N^{-1/2} {\bf 1}^T ({\bf L}^T-{\bf L})= {\bf 0}.$$
This together with
  \eqref{Euleriangraph.thm.pfeq9}, \eqref{Euleriangraph.thm.pfeq10}, \eqref{Euleriangraph.thm.pfeq11} and  \eqref{Euleriangraph.thm.pfeq12}
  completes the proof of \eqref{Euleriangraph.thm.eq2} and
  \eqref{Euleriangraph.thm.eq3}.

\subsection{Proof of Theorem  \ref{Euleriangraph.cor}}\label{Euleriangraph.cor.pfappendix}

By \eqref{Euleriangraph.thm.eq0},  the SVD \eqref{Euleriangraph.eq3}
 for the Laplacian $\L_t, 0\le t\le 1$ is {\em unique}, up to a sign for each eigenvectors, i.e.,
 \begin{equation}\label{Euleriangraph.cor.pfeq1} {\widetilde {\bf U}}_t=  {\bf U}_t {\bf  S}_t \ \ {\rm and} \ \ {\widetilde {\bf V}}_t=   {\bf V}_t
{\bf S}_t \end{equation}
 for any orthogonal pairs  $( {\bf U}_t, {\bf V}_t)$ and
 $(\widetilde {\bf U}_t, \widetilde {\bf V}_t)$ in  the singular value decomposition \eqref{Euleriangraph.eq3},
 where ${\bf S}_t$ is  diagonal matrix with $\pm 1$ as its diagonal entries.
 In our setting, we observe from  the singular value decomposition \eqref{Euleriangraph.eq3} that
 \begin{equation}\label{Euleriangraph.cor.pfeq2}
 {\bf L}_t= {\bf U}_t {\pmb \Sigma}_t {\bf V}_t^T= {\bf V}_{1-t} {\pmb \Sigma}_t {\bf U}_{1-t}^T,\  0\le t\le 1.
 \end{equation}
 By  \eqref{Euleriangraph.cor.pfeq1} and
  \eqref{Euleriangraph.cor.pfeq2}, and orthogonality property for ${\bf U}_t$ and ${\bf V}_t, 0\le t\le 1$, there exists
  diagonal matrices ${\bf S}_t, 0\le t\le 1$, with diagonal entries $\pm 1$ such that
  \begin{equation}  \label{Euleriangraph.cor.pfeq3}
  {\bf V}_{1-t}^T{\bf U}_t= {\bf U}_{1-t}^T{\bf V}_{t}={\bf S}_t.
  \end{equation}
  Recall that $ {\bf U}_t$ and ${\bf V}_t$ are continuous about $t\in [0,1]$.
  This, together with  \eqref{Euleriangraph.cor.pfeq3} and the observation that
    ${\bf S}_t, 0\le t\le 1$, have entries taking values $-1, 0, 1$,
   implies that ${\bf S}_t, 0\le t\le 1$, is independent on $t$, i.e.,
   \begin{equation} \label{Euleriangraph.cor.pfeq4}
   {\bf S}_t={\bf S}_{1/2}, \ 0\le t\le 1.
   \end{equation}
   For $t=1/2$, ${\mathbf L}_t$ is a symmetric matrix with all eigenvalues being simple by
   \eqref{Euleriangraph.thm.eq0}, which implies that
   ${\bf U}_{1/2}={\bf V}_{1/2}$. Hence
 ${\bf S}_{1/2}={\bf I}$.
This together with \eqref{Euleriangraph.cor.pfeq3} and
 and \eqref{Euleriangraph.cor.pfeq4}
proves \eqref{Euleriangraph.cor.eq}.

The relationship \eqref{Euleriangraph.cor.eq+}
between GFTs ${\mathcal F}_{1-t}$ and ${\mathcal F}_t, 0\le t\le 1$, follows directly from
\eqref{jointfourier.def} and \eqref{Euleriangraph.cor.eq}.

\subsection{Proof of Theorem \ref{missundirected.thm}}\label{missundirected.thm.pfsection}

Set ${\bf U}={\bf U}_{t_0}={\bf U}_{t_1}$.
By \eqref{Euleriangraph.eq3} and \eqref{missundirected.thm.eq1}, we have
\begin{eqnarray*}
& \hskip-0.08in  & \hskip-0.08in
((1-t_0){\bf L}^T+t_0 {\bf L}) ((1-t_1){\bf L}+t_1{\bf L}^T)^T\nonumber\\
& \hskip-0.08in = & \hskip-0.08in {\bf L}_{t_0} ({\bf L}_{t_1})^T=
{\bf U} {\pmb \Sigma}_{t_0} {\bf \pmb \Sigma}_{t_1} {\bf U}= {\bf U} {\pmb \Sigma}_{t_1} {\bf \pmb \Sigma}_{t_0} {\bf U}= {\bf L}_{t_1} ({\bf L}_{t_0})^T\nonumber\\
& \hskip-0.08in = & \hskip-0.08in ((1-t_1){\bf L}^T+t_1 {\bf L}) ((1-t_0){\bf L}+t_0{\bf L}^T)^T.
\end{eqnarray*}
Simplifying the above equality and using $t_0\ne t_1$ proves the conclusion that $({\bf L}^T)^2={\bf L}^2$.

    \end{document}